\documentclass[12pt,onecolumn, journal]{IEEEtran}
\usepackage{layout,calc}
\usepackage{times,amssymb,amsmath,epsfig,amsthm}
\usepackage{citesort}
\usepackage{subfigure}
\usepackage{setspace}
\usepackage{url}
\usepackage{cite}
 \usepackage{setspace}
\usepackage[usenames,dvipsnames,svgnames,table]{xcolor}

\usepackage{color}
\newcommand{\hide}[1]{}
\definecolor{gray}{rgb}{0.5,0.5,0.5}

\newtheorem{theorem}{Theorem}
   \newtheorem{lemma}[theorem]{Lemma}
   \newtheorem{corollary}[theorem]{Corollary}
   \newtheorem{definition}{Definition}

\usepackage{subfigure}
\IEEEoverridecommandlockouts
\begin{document}
\title{On the Capacity Region and the Generalized Degrees of Freedom Region for the MIMO Interference Channel with Feedback}
\author{Mehdi Ashraphijuo, Vaneet Aggarwal and Xiaodong Wang \thanks{

The material in this paper was presented in part at the IEEE International Symposium on Information Theory, Istanbul, Turkey, July 2013.

M. Ashraphijuo and X. Wang are with the Electrical Engineering Department, Columbia University, New York, NY 10027, email: \{mehdi,wangx\}@ee.columbia.edu. V. Aggarwal is with AT\&T Labs-Research, Florham Park, NJ 07932, email: vaneet@research.att.com}}

\maketitle
\begin{abstract}
In this paper, we study the effect of feedback on the two-user MIMO interference channel. The capacity region of the MIMO interference channel with feedback is characterized within a constant number of bits, where this constant is independent of the channel matrices. Further, it is shown that the capacity region of the MIMO interference channel with feedback and its reciprocal interference channel are within a constant number of bits. Finally, the generalized degrees of freedom region for the MIMO interference channel with feedback is characterized.
\end{abstract}

\newpage

\section{Introduction}
Wireless networks with multiple users are interference-limited rather than noise-limited. The interference channel (IC) is a good starting point for understanding the performance limits of the interference limited communications \cite{Etkin,araja,sason,Varanasi,gdof,gdof.tarokh,Jafar}. Feedback can be employed in the ICs to achieve an improvement in the data rates \cite{Tse,Tuninetti1,Tuninetti2,Achal1st,Achal,Alireza}. However, most of the existing works on the ICs with feedback are limited to discrete memoryless channels, or the single-input single-output (SISO) channels. This paper analyzes the multiple-input multiple-output (MIMO) Gaussian IC with feedback. 

{
In this paper, we consider the two-user MIMO IC with perfect channel state knowledge at the transmitters and receivers. In large wireless networks, having global knowledge of the channel state is infeasible and thus the authors of \cite{Lozano} found a saturation effect in the system capacity.  In this paper, we assume that all the nodes know the channel state information of all the links to find the impact of feedback to the transmitters, which is a fundamental question on its own. While the overhead of gathering global channel state information must not be neglected, it has been repeatedly shown (cf. \cite{Caire1,Caire2}) that this overhead is manageable in the presence of a reduced number of users. This overhead increases as the number of users increases, and thus some authors have considered knowledge of channel state in a local neighborhood \cite{Vaneet2,Kanes}. With the local network connectivity and channel state information, sub-networks can be scheduled where each sub-network is operated using an information-theoretic optimal scheme \cite{Pedro,Pedro2}. Thus, even with the knowledge of the local channel state information, understanding of small networks can help improve throughput of large networks.}
{


Finding a capacity achieving scheme for an IC with more than two users is an open problem, and assumptions like treating interference as noise have been used \cite{Lozano,kkkkk,kkkkkKK}. An approximate capacity region  for the two-user SISO IC was given in \cite{Etkin}, which has been further extended to the MIMO IC in \cite{Varanasi}. Even an approximate capacity region is an open problem beyond two-user IC, although capacity regions have been found in some special cases like double-Z \cite{vaneet}, one-to-many \cite{bresler}, many-to-one \cite{bresler}, and cyclic \cite{cyc} ICs. In the presence of feedback, an approximate capacity region for the two-user SISO IC was recently given in \cite{Tse}, where the capacity region is characterized within two bits.} It was shown that the capacity regions of Gaussian ICs increase unboundedly with feedback unlike the Gaussian multiple-access channel where the gains are bounded \cite{ozarow}. The degrees of freedom for a symmetric SISO Gaussian IC with feedback is also found in \cite{Tse}. In this paper, we find an outer bound and an inner bound for the capacity region that differ by a constant number of bits, and also evaluate the generalized degrees of freedom (GDoF) region for a general MIMO IC with feedback.

The first main result of the paper is the characterization of the capacity region of a MIMO IC with feedback within $N_1+N_2+\max(N_1,N_2)$ bits, where $N_1$ and $N_2$ are the numbers of receive antennas at the two receivers. An outer-bound is obtained by first outer bounding the covariance matrices of both input signals and representing the outer bound as a region in terms of the covariance matrix between the two input signals. This is further outer-bounded by a larger region that does not involve the covariance matrix. The achievability strategy is  based on block Markov encoding, backward decoding, and Han-Kobayashi message-splitting. This achievable rate and the outer bound are within  $N_1+N_2+\max(N_1,N_2)$ bits of each other thus characterizing the capacity region of the two-user IC within constant number of bits where the constant is independent of the channel matrices. The achievability scheme that is used to prove the constant gap result assumes that the transmitted signals from the two transmitters in a time-slot are uncorrelated, unlike \cite{Tse} where the signals were assumed correlated in the achievability. Thus, our achievable rate region is within $3$ bits rather than $2$ bits as in \cite{Tse} of the capacity region of a SISO IC with feedback. An achievability scheme without correlated inputs was also shown to achieve within constant gap of the capacity region in \cite{Achal} for a SISO IC with feedback. However, our gap between the inner and the outer bounds is smaller as compared to \cite{Achal}.

We note that the achievability strategies for a SISO IC in \cite{Tse,Achal} emphasize that the private part from a transmitter using the Han-Kobayashi message splitting is such that it is received at the other receiver at the noise floor. However for a MIMO IC with feedback, it is not clear what its counterpart would be. The Han-Kobayashi message splitting used in this paper gives the notion of receiving the signal at the noise floor for a MIMO IC with feedback. Many matrix based results are derived in this paper to show a constant gap between the outer and the inner bounds of the capacity region of a MIMO IC with feedback, which may be of independent interest.

The second main result of the paper is to show that the capacity region of a MIMO IC with feedback and that of its corresponding reciprocal channel are within constant number of bits of each other, where the constant is independent of channel matrices. The reciprocal IC was considered in \cite{Varanasi}, where the authors showed that the capacity region of a MIMO IC without feedback is within constant number of bits of its corresponding reciprocal IC. This paper shows that the constant gap between a MIMO IC and its reciprocal channel also holds in the presence of feedback.

{
Most developments on the IC take place in the high-power regime, and the GDoF region characterizes the capacity region in the limit of high-power. Thus, we further extend our results to high power regime to get more understanding on the improvement in the capacity region with feedback. The GDoF region has been characterized in the symmetric case without feedback \cite{K3} and with feedback \cite{Moh} for a $K$-user SISO IC. For a general MIMO IC without feedback, the GDoF region is found for a two-user IC in \cite{gdof}.}

The third main result of the paper is a complete characterization of the GDoF region of a general MIMO IC with feedback when the average signal quality of each link, say $\rho_{ij}$ for link from transmitter $i$ to receiver $j$, varies with a  base signal-to-noise ratio (SNR) parameter, say $\mathsf{SNR}$, as $\lim_{\mathsf{SNR}\to\infty} \frac{\log{\rho_{ij}}}{\log \mathsf{SNR}} = \alpha_{ij}$, where $\alpha_{ij}$ can be different for each link with $i,j\in \{1,2\}$. In other words, the average link quality of each link can potentially have different exponents of a base SNR. As a special case, we consider a symmetric IC where the number of antennas at both transmitters is the same, the number of antennas at both receivers is the same, and the SNRs for the direct links and the cross links are $\mathsf{SNR}$ and $\mathsf{SNR}^\alpha$, $\alpha\ge 0$, respectively. We find the GDoF (the maximum symmetric point in the GDoF region) for a given $\alpha$ and show that the GDoF is a ``V"-curve rather than a ``W"-curve corresponding to the GDoF without feedback as in \cite{gdof}. Similar result was obtained for a SISO IC in \cite{Tse} while this paper extends it to a MIMO system.

The remainder of the paper is organized as follows. Section II introduces the model for a MIMO IC with feedback, reciprocal IC and the GDoF region. Sections III and IV describe our results on the capacity region and the GDoF region respectively. Section V concludes the paper. The detailed proofs of various results are given in Appendices A-E.


\section{Channel Model and Preliminaries}
In this section, we describe the channel model considered in this paper. A two-user MIMO IC consists of two transmitters and two receivers. Transmitter $i$ is labeled as $\mathsf{T}_i$ and receiver $j$ is labeled as $\mathsf{D}_j$ for $i,j \in \{1,2\}$. Further, we assume $\mathsf{T}_i$ has $M_i$ antennas and $\mathsf{D}_i$ has $N_i$ antennas, $i \in \{1,2\}$. Henceforth, such a MIMO IC will be referred to  as the $(M_1,N_1,M_2,N_2)$ MIMO IC. We assume that the channel matrix between transmitter $\mathsf{T}_i$ and receiver $\mathsf{D}_j$ is denoted by  $H_{ij}\in \mathbb{C}^{N_j\times M_i}$, for $i,j \in \{1,2\}$. We shall consider a time-invariant or fixed channel where the channel matrices remain fixed for the entire duration of communication. At each discrete time instance, indexed by $t=1, 2, \cdots$, transmitter $\mathsf{T}_i$ transmits a vector $X_{i}[t]\in \mathbb{C}^{M_i\times 1}$ over the channel with a power constraint ${{\rm tr}(\mathbb{E}(X_{i}X^{\dagger }_{i}))}\le 1$ ($A^\dagger$ denotes the conjugate transpose of the matrix $A$).

Let $Q_{ij}=\mathbb{E}(X_{i}X^{\dagger }_{j})$ for $i,j\in \{1,2\}$. We say $A\preceq B$ if $B-A$ is a positive semi-definite (p.s.d.) matrix and we say $A\succeq B$ if $B\preceq A$. The identity matrix of size $s\times s$ is denoted by $I_{s}$. Further, we define $x^+\triangleq \max\{x,0\}$. We also note that $0 \preceq Q_{ii}\preceq I$ according to Theorem $7.7.3.$ of \cite{Q} since ${{\rm tr}(\mathbb{E}(X_{i}X^{\dagger }_{i}))}\le 1$. By definition of $Q_{ij}$, we see that $Q_{ij}=Q_{ji}^\dagger$. Moreover, we have $0\preceq Q_{ij}Q_{ij}^\dagger\preceq I$, where $0\preceq Q_{ij}Q_{ij}^\dagger$ results from the fact that every matrix in the form of $AA^{\dagger}$ is p.s.d. and $Q_{ij}Q_{ij}^\dagger\preceq I$ results from ${{\rm tr}(Q_{ij}{Q_{ij}}^{\dagger})}={{\rm tr}(Q_{ii})}{{\rm tr}(Q_{jj})}\le1$ which gives $Q_{ij}{Q_{ij}}^{\dagger}\preceq I$ with a similar argument as we had for $Q_{ii}$.
We will sometimes denote $Q=Q_{12}$ when it does not lead to confusion.

We also incorporate a non-negative power attenuation factor, denoted as ${\rho}_{ij}$, for the signal transmitted from $\mathsf{T}_i$ to $\mathsf{D}_j$. The received signal at receiver $\mathsf{D}_i$ at discrete time instance $t$ is denoted as $Y_{i}[t]$ for $i\in \{1,2\}$, and can be written as
\begin{eqnarray}
Y_{1}[t]&=&\sqrt{{\rho }_{11}}H_{11}X_{1}[t]+\sqrt{{\rho }_{21}}H_{21}X_{2}[t]+Z_{1}[t],\\
Y_{2}[t]&=&\sqrt{{\rho }_{12}}H_{12}X_{1}[t]+\sqrt{{\rho }_{22}}H_{22}X_{2}[t]+Z_{2}[t],
\end{eqnarray}
where $Z_{i}[t]\in \mathbb{C}^{N_i\times 1}$ is independent and identically distributed (i.i.d.) $\mathsf{CN}(0,I_{N_i})$ (complex Gaussian noise), ${\rho }_{ii}$ is the received SNR at $\mathsf{D}_i$ and ${\rho }_{ij}$ is the received interference-to-noise-ratio at $\mathsf{D}_j$ for $i, j\in \{1,2\},$ $i \ne j$. A MIMO IC is fully described by three parameters. The first is the number of antennas at each transmitter and receiver, namely $(M_1,N_1,M_2,N_2)$. The second is the set of channel gains, $\overline{H}=\{H_{11},H_{12},H_{21},H_{22}\}$. The third is the set of average link qualities of all the channels,  ${\overline{\rho }}=\{{\rho }_{11}{,{\rho }_{12},\rho }_{21},{\rho }_{22}\}$. We assume that these parameters are known to all transmitters and receivers.

For MIMO IC with feedback, the transmitted signal $X_{i}[t]$ at $\mathsf{T}_i$ is a function of the message $W_i$ and the previous channel outputs at $\mathsf{D}_i$ for $i \in \{1,2\}$. Thus, the encoding functions of the two transmitters are given as
\begin{eqnarray}
X_i[t]&=&f_{it}(W_i,Y^{t-1}_i), \ \ \ i \in \{1,2\},
\end{eqnarray}
where $f_{it}$ is the encoding function of $\mathsf{T}_i$, $W_i$ is the message of $\mathsf{T}_i$ and $Y^{t-1}_i=(Y_i[1],..., Y_i[t-1])$. Similarly, we denote $X_i^{t}=(X_i[1],..., X_i[t])$. Let us assume that $\mathsf{T}_i$ transmits information at a rate of $R_i$ to $\mathsf{D}_i$ using the codebook $C_{i,n}$ of length-$n$ codewords with $|C_{i,n}|=2^{nR_i}$. Given a message $m_i\in \{1,\dots ,2^{nR_i}\}$, the corresponding codeword $X_i^n\in C_{i,n}$ satisfies the power constraint mentioned before. From the received signal $Y^n_i$, the receiver obtains an estimate $\widehat{m_i}$ of the transmitted message $m_i$ using a decoding function. Let the average probability of error be denoted by $e_{i,n}={\rm Pr}({\rm \ }\widehat{m_i}\ne m_i)$.

A rate pair $(R_1,R_2)$ is achievable if there exists a family of codebooks $C_{i,n}$ and decoding functions such that ${{\rm max}}_i\{e_{i,n}\}$ goes to zero as the block length $n$ goes to infinity. The capacity region $\mathbb{C}(\overline{H},\overline{\rho })$ of the IC with parameters $\overline{H}$ and $\overline{\rho }$ is defined as the closure of the set of all achievable rate pairs.

Consider a two-dimensional rate region $\mathbb{C}$. Then, the region $\mathbb{C}\oplus ([0,a]\times [0,b])$ denotes the region formed by $\{(R_1,R_2): R_1, R_2\ge 0, ((R_1-a)^+, (R_2-b)^+)\in \mathbb{C}\}$ for some $a,b\ge 0$. Similarly, the region $\mathbb{C}\ominus ([0,a]\times [0,b])$ denotes the region formed by $\{(R_1,R_2): R_1, R_2\ge 0, ((R_1+a)^+, (R_2+b)^+)\in \mathbb{C}\}$  for some $a,b\ge 0$. Further, we define the notion of an achievable rate region that is within a constant number of bits of the capacity region as follows.

\begin{definition}
An achievable rate region $\mathbb{A}(\overline{H},\overline{\rho })$ is said to be within $b$ bits of the capacity region if $\mathbb{A}(\overline{H},\overline{\rho }) \subseteq \mathbb{C}(\overline{H},\overline{\rho })$ and $\mathbb{A}(\overline{H},\overline{\rho })\oplus ([0,b]\oplus[0,b]) \supseteq \mathbb{C}(\overline{H},\overline{\rho })$.
\end{definition}

In this paper, we will use the GDoF region to characterize the capacity region of the MIMO IC with feedback in the limit of high SNR. This notion generalizes the  conventional degrees of freedom (DoF) region metric by additionally emphasizing the signal level as a signaling dimension. It characterizes the simultaneously accessible fractions of spatial and signal-level dimensions (per channel use) by the two users when all the average channel coefficients vary as exponents of a nominal SNR parameter. Thus, we assume that
\begin{eqnarray}
{\mathop{{\lim}}_{{\log  \mathsf{SNR} \to \infty \ }} \frac{{\log  {\rho }_{ij}\ }}{{\log  \mathsf{SNR} \ }}\ }={\alpha }_{ij},
\end{eqnarray}
where ${\alpha }_{ij}\in {\mathbb R}^+$ for all $i,j\in \{1,\ 2\}$. In the limit of high SNR, the capacity region diverges.

The GDoF region is defined as the region formed by the set of all $(d_1,d_2)$ such that $(d_1 \log\mathsf{SNR}-o(\log\mathsf{SNR}), d_2 \log\mathsf{SNR}-o(\log\mathsf{SNR}))$\footnote{$a=o(\mathsf{\log SNR})$ indicates that $\lim_{\mathsf{SNR}\to\infty}\frac{a}{\log\mathsf{SNR}}=0$.} is inside the capacity region. Thus, the GDoF is a function of link quality scaling exponents $\alpha_{ij}$. We note that since the channel matrices are of full ranks with probability 1, we will have the GDoF with probability 1 over the randomness of channel matrices.

%

The property of maintaining the same performance even if the direction of information flow is reversed is known as the \textit{reciprocity} of the channel. For a MIMO IC with parameters $(M_1,N_1,M_2,N_2)$, $\overline{H}=\{H_{11},H_{12},H_{21},H_{22}\}$, and ${\overline{\rho }}=\{{\rho }_{11}{,{\rho }_{12},\rho }_{21},{\rho }_{22}\}$, the reciprocal MIMO IC has parameters $(N_1,M_1,N_2,M_2)$, $\overline{H}^R=\{H^T_{11},H^T_{21},H^T_{12},H^T_{22}\}$, and ${\overline{\rho }}^R=\{\rho_{11},\rho_{21},\rho_{12},\rho_{22}\}$.


%

\section{Capacity Region of MIMO IC with Feedback}\label{main_results}
In this section, we will describe our results on the capacity region of the two-user MIMO IC with feedback.

Our first result gives an outer bound on the capacity region of the two-user MIMO IC with feedback. Let $\mathcal{R}_o(Q)$ be the region formed by $(R_1,R_2)$ satisfying the following constraints for some covariance matrix $Q=\mathbb{E}[X_1X^{\dagger }_2]$:

\begin{eqnarray}
R_1&\le& {\log  {\det  (I_{N_1}+{\rho }_{11}H_{11}H^{\dagger }_{11}+{\rho }_{21}H_{21}H^{\dagger }_{21}+\sqrt{{\rho }_{11}{\rho }_{21}}H_{11}QH^{\dagger }_{21}+\sqrt{{\rho }_{11}{\rho }_{21}}H_{21}Q^{\dagger }H^{\dagger }_{11})}},\label{roqeq1}\\
R_2&\le& {\log  {\det  (I_{N_2}+{\rho }_{22}H_{22}H^{\dagger }_{22}+{\rho }_{12}H_{12}H^{\dagger }_{12}+\sqrt{{\rho }_{22}{\rho }_{12}}H_{22}Q^{\dagger }H^{\dagger }_{12}+\sqrt{{\rho }_{22}{\rho }_{12}}H_{12}QH^{\dagger }_{22})}},\\
R_1&\le& \log  \det  \left(I_{N_2}+{\rho }_{12}H_{12}H^{\dagger }_{12}-{\rho }_{12}H_{12}QQ^{\dagger }H^{\dagger }_{12}\right)+\log  \det \Bigg(I_{N_1}+{\rho }_{11}H_{11}H^{\dagger }_{11}-\nonumber\\
&&\left[ \begin{array}{cc}
\sqrt{{\rho }_{11}{\rho }_{12}}H_{11}H^{\dagger }_{12} & \sqrt{{\rho }_{11}}H_{11}Q \end{array}
\right]
{\left[ \begin{array}{cc}
I_{N_2}+{\rho }_{12}H_{12}H^{\dagger }_{12} & \sqrt{{\rho }_{12}}H_{12}Q \\
\sqrt{{\rho }_{12}}Q^{\dagger }H^{\dagger }_{12} & I_{M_2} \end{array}
\right]}^{-1}\nonumber\\
&&\left[ \begin{array}{c}
\sqrt{{\rho }_{11}{\rho }_{12}}H_{12}H^{\dagger }_{11} \\
\sqrt{{\rho }_{11}}Q^{\dagger }H^{\dagger }_{11} \end{array}
\right]\Bigg),\\
R_2&\le& \log  \det  \left(I_{N_1}+{\rho }_{21}H_{21}H^{\dagger }_{21}-{\rho }_{21}H_{21}Q^{\dagger }QH^{\dagger }_{21}\right)
+\log \det  \Bigg(I_{N_2}+{\rho }_{22}H_{22}H^{\dagger }_{22}-\nonumber\\
&&\left[ \begin{array}{cc}
\sqrt{{\rho }_{22}{\rho }_{21}}H_{22}H^{\dagger }_{21} & \sqrt{{\rho }_{22}}H_{22}Q^{\dagger } \end{array}
\right]{\left[ \begin{array}{cc}
I_{N_1}+{\rho }_{21}H_{21}H^{\dagger }_{21} & \sqrt{{\rho }_{21}}H_{21}Q^{\dagger } \\
\sqrt{{\rho }_{21}}QH^{\dagger }_{21} & I_{M_1} \end{array}
\right]}^{-1}\nonumber\\
&&\left[ \begin{array}{c}
\sqrt{{\rho }_{22}{\rho }_{21}}H_{21}H^{\dagger }_{22} \\
\sqrt{{\rho }_{22}}QH^{\dagger }_{22} \end{array}
\right]\Bigg),\\
R_1+R_2&\le& \log \det  \left(I_{N_2}+{\rho }_{22}H_{22}H^{\dagger }_{22}+{\rho }_{12}H_{12}H^{\dagger }_{12}+\sqrt{{\rho }_{22}{\rho }_{12}}H_{22}Q^{\dagger }H^{\dagger }_{12}+\sqrt{{\rho }_{22}{\rho }_{12}}H_{12}QH^{\dagger }_{22}\right)\nonumber\\
&&+\log\det  \Bigg(I_{N_1}+{\rho }_{11}H_{11}H^{\dagger }_{11}-\left[ \begin{array}{cc}
\sqrt{{\rho }_{11}{\rho }_{12}}H_{11}H^{\dagger }_{12} & \sqrt{{\rho }_{11}}H_{11}Q \end{array}
\right]\nonumber\\
&&{\left[ \begin{array}{cc}
I_{N_2}+{\rho }_{12}H_{12}H^{\dagger }_{12} & \sqrt{{\rho }_{12}}H_{12}Q \\
\sqrt{{\rho }_{12}}Q^{\dagger }H^{\dagger }_{12} & I_{M_2} \end{array}
\right]}^{-1}\left[ \begin{array}{c}
\sqrt{{\rho }_{11}{\rho }_{12}}H_{12}H^{\dagger }_{11} \\
\sqrt{{\rho }_{11}}Q^{\dagger }H^{\dagger }_{11} \end{array}
\right]\Bigg),\\
R_1+R_2&\le& \log  \det  \left(I_{N_1}+{\rho }_{11}H_{11}H^{\dagger }_{11}+{\rho }_{21}H_{21}H^{\dagger }_{21}+\sqrt{{\rho }_{11}{\rho }_{21}}H_{11}QH^{\dagger }_{21}+\sqrt{{\rho }_{11}{\rho }_{21}}H_{21}Q^{\dagger }H^{\dagger }_{11}\right)\nonumber\\
&&+\log \det  \Bigg(I_{N_2}+{\rho }_{22}H_{22}H^{\dagger }_{22}-\left[ \begin{array}{cc}
\sqrt{{\rho }_{22}{\rho }_{21}}H_{22}H^{\dagger }_{21} & \sqrt{{\rho }_{22}}H_{22}Q^{\dagger } \end{array}
\right]\nonumber\\
&&{\left[ \begin{array}{cc}
I_{N_1}+{\rho }_{21}H_{21}H^{\dagger }_{21} & \sqrt{{\rho }_{21}}H_{21}Q^{\dagger } \\
\sqrt{{\rho }_{21}}QH^{\dagger }_{21} & I_{M_1} \end{array}
\right]}^{-1}\left[ \begin{array}{c}
\sqrt{{\rho }_{22}{\rho }_{21}}H_{21}H^{\dagger }_{22} \\
\sqrt{{\rho }_{22}}QH^{\dagger }_{22} \end{array}
\right]\Bigg).\label{roqeql}
\end{eqnarray}
Further, let $\mathcal{R}_o$ be the convex hull of $\mathcal{R}_o(Q)$ for all covariance matrices $Q$. The following theorem outer bounds the capacity region of the two-user MIMO IC with feedback.

\begin{theorem}
The capacity region of the two-user MIMO IC with perfect feedback $\mathbb{C}_{FB}$ is bounded from above as follows
\begin{equation}
\mathbb{C}_{FB}\subseteq \mathcal{R}_o.
\end{equation}
\label{outer_capacity}
\end{theorem}
\begin{proof}
The proof is given in Appendix \ref{apdx_outer}.
\end{proof}

From the definition of $\mathcal{R}_o(Q)$, by substituting $Q=0$ and after some simplifications, we get that $\mathcal{R}_o(0)$ is the region formed by $(R_1,R_2)$ satisfying the following
\begin{eqnarray}
R_1&\le& {\log  {\det  (I_{N_1}+{\rho }_{11}H_{11}H^{\dagger }_{11}+{\rho }_{21}H_{21}H^{\dagger }_{21})}},\label{ro0eq1}\\
R_2&\le& {\log  {\det  (I_{N_2}+{\rho }_{22}H_{22}H^{\dagger }_{22}+{\rho }_{12}H_{12}H^{\dagger }_{12})}},\\
R_1&\le& \log  \det  \left(I_{N_2}+{\rho }_{12}H_{12}H^{\dagger }_{12}\right)+\log  \det (I_{N_1}+{\rho }_{11}H_{11}H^{\dagger }_{11}-\nonumber\\
&& \sqrt{{\rho }_{11}{\rho }_{12}}H_{11}H^{\dagger }_{12}
({I_{N_2}+{\rho }_{12}H_{12}H^{\dagger }_{12}})^{-1}
\sqrt{{\rho }_{11}{\rho }_{12}}H_{12}H^{\dagger }_{11}),\\
R_2&\le& \log  \det  \left(I_{N_1}+{\rho }_{21}H_{21}H^{\dagger }_{21}\right)
+\log \det (I_{N_2}+{\rho }_{22}H_{22}H^{\dagger }_{22}-\nonumber\\
&&\sqrt{{\rho }_{22}{\rho }_{21}}H_{22}H^{\dagger }_{21}
({I_{N_1}+{\rho }_{21}H_{21}H^{\dagger }_{21}})^{-1}
\sqrt{{\rho }_{22}{\rho }_{21}}H_{21}H^{\dagger }_{22}),\\
R_1+R_2&\le& \log \det  \left(I_{N_2}+{\rho }_{22}H_{22}H^{\dagger }_{22}+{\rho }_{12}H_{12}H^{\dagger }_{12}\right)
+\log  \det (I_{N_1}+{\rho }_{11}H_{11}H^{\dagger }_{11}-\nonumber\\
&& \sqrt{{\rho }_{11}{\rho }_{12}}H_{11}H^{\dagger }_{12}
({I_{N_2}+{\rho }_{12}H_{12}H^{\dagger }_{12}})^{-1}
\sqrt{{\rho }_{11}{\rho }_{12}}H_{12}H^{\dagger }_{11}),\\
R_1+R_2&\le& \log  \det  \left(I_{N_1}+{\rho }_{11}H_{11}H^{\dagger }_{11}+{\rho }_{21}H_{21}H^{\dagger }_{21}\right)
+\log \det (I_{N_2}+{\rho }_{22}H_{22}H^{\dagger }_{22}-\nonumber\\
&&\sqrt{{\rho }_{22}{\rho }_{21}}H_{22}H^{\dagger }_{21}
({I_{N_1}+{\rho }_{21}H_{21}H^{\dagger }_{21}})^{-1}
\sqrt{{\rho }_{22}{\rho }_{21}}H_{21}H^{\dagger }_{22}).\label{ro0eql}
\end{eqnarray}

The following result gives an inner bound to the capacity region of the two-user MIMO IC with feedback.

\begin{theorem}
The capacity region for the two-user MIMO IC with perfect feedback $\mathbb{C}_{FB}$ is bounded from below as
\begin{equation}
\mathbb{C}_{FB} \supseteq  \mathcal{R}_o(0) \ominus ([0,N_1+N_2]\times [0,N_1+N_2]).
\end{equation}
\label{inner_capacity}
\end{theorem}
\begin{proof}
The proof is provided in Appendix \ref{apdx_inner}.
\end{proof}

The inner bound uses the achievable region for a two-user discrete memoryless IC with feedback as in \cite{Tse}. The achievability scheme employs block Markov encoding, backward decoding, and Han-Kobayashi message-splitting. This result for a discrete memoryless channel is extended to MIMO IC with feedback using a specific message splitting by power allocation. The transmitted signal $X_i$ from $\mathsf{T}_i$ is given as
\begin{eqnarray}
X_i=X_{ip}+X_{iu},
\end{eqnarray}
where $X_{ip}$ and $X_{iu}$ denote the private and public messages of $\mathsf{T}_i$, respectively. We assume that $X_{ip}$ and $X_{iu}$ are independent for $i=1,2$. However, these transmitted signals are correlated over time due to block Markov encoding. The private signal $X_{ip}$ is chosen to be $X_{ip}\sim \mathsf{CN}\left(0,K_{X_{ip}}\right)$, and the public signal $X_{iu}$ is chosen to be $X_{iu}\sim \mathsf{CN}\left(0,K_{X_{iu}}\right)$, where
\begin{eqnarray}\label{ipi}
K_{X_{ip}}
=I_{M_i}-
\sqrt{{\rho }_{ij}}H^{\dagger }_{ij}
(I_{N_j}+{\rho }_{ij}H_{ij}H^{\dagger }_{ij})^{-1}
\sqrt{{\rho }_{ij}}H_{ij},
\end{eqnarray}
and
\begin{eqnarray}\label{iui}
K_{X_{iu}}=I_{M_i}-K_{X_{ip}},
\end{eqnarray}
for $i\in \{1,2\}$.

We will show in Appendix \ref{apdx_inner} that the power allocation is feasible by showing $K_{X_{ip}}\succeq 0$ and $K_{X_{iu}}\succeq 0$. Further, this message split is such that the private signal is received at the other receiver with power bounded by a constant. More specifically we have $\rho_{ij}H_{ij}K_{X_{ip}} H_{ij}^\dagger \preceq I_{N_j}$, thus showing that the effective received signal covariance matrix at $\mathsf{D}_j$ corresponding to the private signal from $\mathsf{T}_i$ is at or below the noise floor.

This power allocation is different from that given in \cite{Tse} even for a SISO channel. Note that the power split levels in the achievability scheme of \cite{Tse} do not sum to $1$ and thus do not satisfy the total power constraint. For the special case of SISO IC with feedback, the above gives a fix to the results in \cite{Tse}. This power allocation assumes uncorrelated signals transmitted by the two users at each time-slot. The authors of \cite{Achal} also used uncorrelated signals for SISO but had a larger gap between the inner and outer bounds for SISO IC with feedback than that achieved by our achievability strategy.



Having considered the inner and outer bounds for the capacity region of the two-user IC with feedback, the next result shows that the inner bound and the outer bound are within $N_1+N_2+\max(N_1,N_2)$ bits thus finding the capacity region of the two-user IC with feedback, approximately.

\begin{theorem}
The capacity region for the two-user MIMO IC with perfect feedback $\mathbb{C}_{FB}$ is bounded from above and below as
 \begin{equation}
 \mathcal{R}_o(0) \ominus ([0,N_1+N_2]\times [0,N_1+N_2]) \subseteq \mathbb{C}_{FB}\subseteq \mathcal{R}_o(0)\oplus ([0,N_1]\times [0,N_2]),
 \end{equation}
 where the inner and outer bounds are within $N_1+N_2+\max{(N_1,N_2)}$ bits. \label{outer_inner_capacity_reciprocal}
 \end{theorem}

\begin{proof}
The inner bound follows from Theorem \ref{inner_capacity}. For outer bound, we outer-bound the region $\mathcal{R}_o(Q)$ as  $\mathcal{R}_o(Q)\subseteq \mathcal{R}_o(0)\oplus ([0,N_1]\times [0,N_2])$ in Appendix \ref{Appendix4}. Hence, $\mathcal{R}_o\subseteq \mathcal{R}_o(0)\oplus ([0,N_1]\times [0,N_2])$. Thus, using $Q=0$ in $\mathcal{R}_o(Q)$ gives an approximate capacity region with the approximation gap as in the statement of the theorem.
\end{proof}

The authors of \cite{Tse} found the capacity region for the SISO IC with feedback within 2 bits. The above theorem generalizes the result to find the capacity region of MIMO IC with feedback within $N_1+N_2+\max(N_1,N_2)$ bits. Note that the approximate capacity region without feedback in \cite{Varanasi} involves bounds on $2R_1+R_2$ which do not appear in our approximate capacity region with feedback. In addition, in \cite{Tse}, the approximate capacity region for the SISO IC with feedback involves the covariance matrix of the inputs in the inner and outer bounds, whereas our approximate capacity region for the MIMO IC with feedback does not.

Figure \ref{fig:uib} gives a pictorial representation for the result of Theorem \ref{outer_inner_capacity_reciprocal}. The inner and the outer bounds for the capacity region for MIMO IC with feedback are within a constant number of bits from the region $\mathcal{R}_o(0)$ and thus the inner and outer bound regions are within a constant number of bits of each other.

\begin{figure}[ht]
\centering{
	\includegraphics[width=8.5cm]{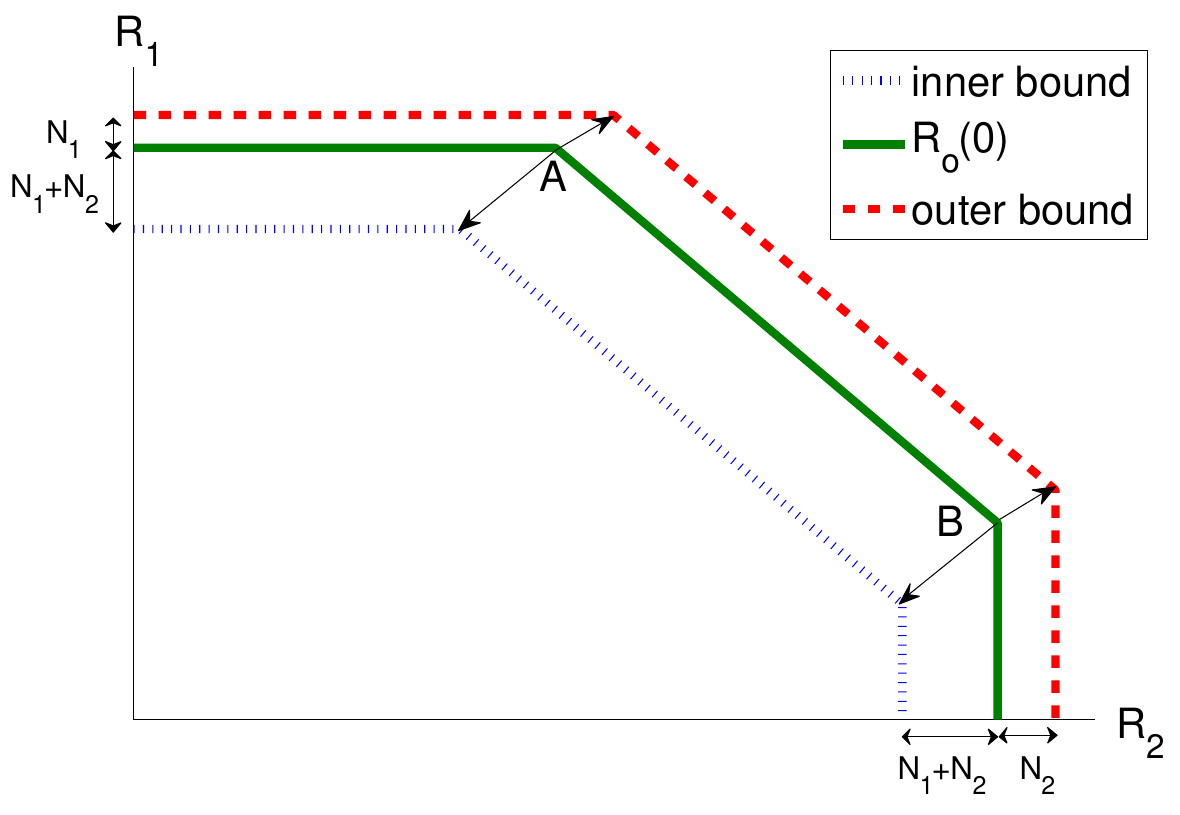}

}

\caption[Optional caption for list of figures]{Inner and outer bounds for the capacity region of MIMO IC with feedback are within a constant number of bits. The arrows from the corners $A$ and $B$ in $\mathcal{R}_o(0)$ toward their respective corners on outer bound have vertical length of $N_1$ and horizontal length of $N_2$. The arrows from the corners $A$ and $B$ in $\mathcal{R}_o(0)$ toward their respective corners on inner bound have the vertical and horizontal length of $N_1+N_2$ each.
}
\label{fig:uib}
\end{figure}

\begin{figure}[ht]
\centering{
	\includegraphics[width=8.5cm]{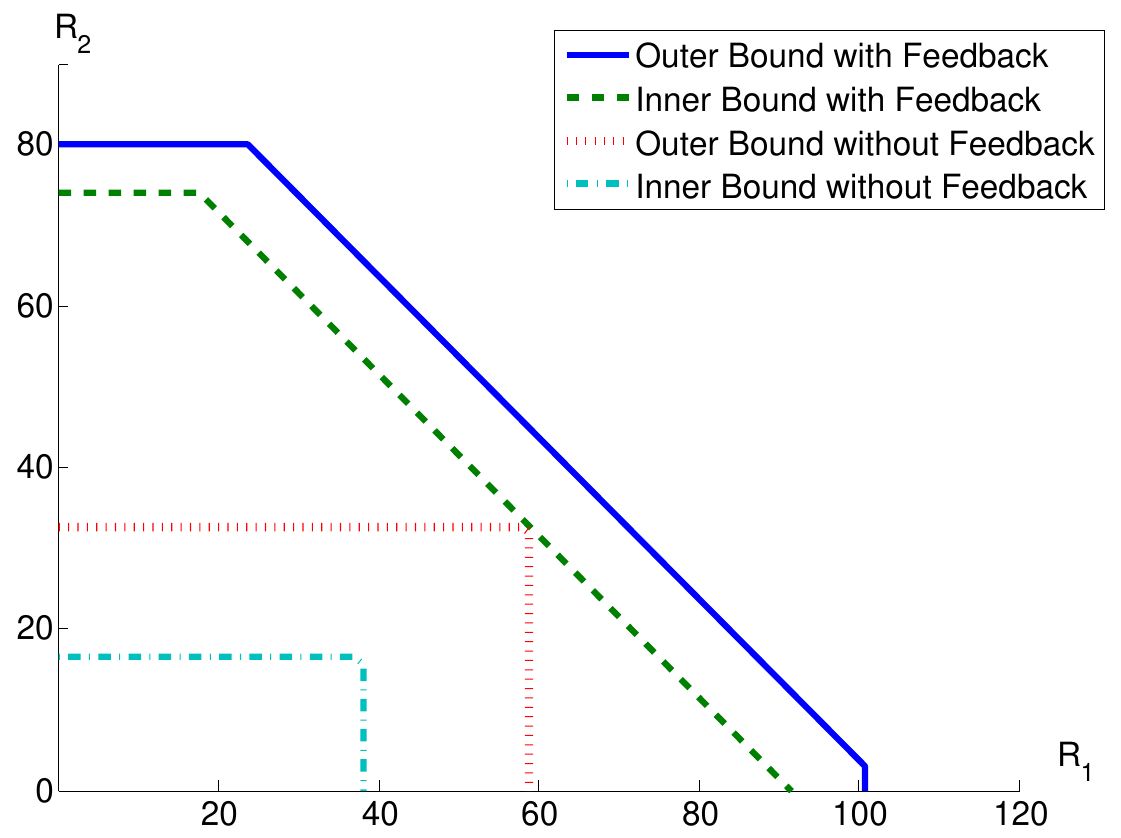}
}
\caption[Optional caption for list of figures]{Inner and outer bounds for the capacity region of MIMO IC with feedback and without feedback.}
\label{fig:uibmm}
\end{figure}

{In Figure \ref{fig:uibmm}, we see the improvement in the capacity region for a MIMO IC with feedback. The parameters chosen for the IC are $M_1=5$, $M_2=4$, $N_1=6$, $N_2=3$, $\rho_{11}=\rho_{22}=10^4$, $\rho_{12}=\rho_{21}=10^{8}$,
\begin{eqnarray}\label{hhhhhh}
H_{11}=\left[\begin{array}{ccccc}
0.30 & 0.19 & 0.10 & 0.68 & 0.65\\
0.30 & 0.44 & 0.38 & 0.60 & 0.94\\
0.35 & 0.65 & 0.98 & 0.58 & 0.65\\
0.56 & 0.14 & 0.82 & 0.92 & 0.72\\
0.28 & 0.42 & 0.19 & 0.39 & 0.28\\
0.46 & 0.89 & 0.49 & 0.20 & 0.72
\end{array}\right],
H_{22}=\left[\begin{array}{cccc}
0.97 & 0.67 & 0.67 & 0.65\\
0.60 & 0.94 & 0.51 & 0.53\\
0.44 & 0.67 & 0.50 & 0.36
\end{array}\right],\nonumber\\
H_{21}=\left[\begin{array}{cccc}
0.89 & 0.95 & 0.41 & 0.69\\
0.81 & 0.59 & 0.65 & 0.98\\
0.61 & 0.44 & 0.60 & 0.37\\
0.82 & 0.16 & 0.83 & 0.72\\
0.10 & 0.82 & 0.92 & 0.28\\
0.87 & 0.43 & 0.91 & 0.21
\end{array}\right], \text{ and }
H_{12}=\left[\begin{array}{ccccc}
0.11 & 0.71 & 0.61 & 0.31 & 0.30\\
0.61 & 0.23 & 0.61 & 0.44 & 0.31\\
0.48 & 0.71 & 0.27 & 0.61 & 0.61
\end{array}\right].
\end{eqnarray}
The inner and outer bounds without feedback are taken from \cite{Varanasi}. We note that the inner bound with feedback contains the outer bound without feedback.}

Having characterized the approximate capacity region for the MIMO IC with feedback, we next explore the relation of capacity region of the MIMO IC with feedback with that of the corresponding reciprocal MIMO IC with feedback. The next theorem shows that the capacity region of the MIMO IC with feedback is approximately the same as that of its corresponding reciprocal channel with feedback.



 \begin{theorem}
 The capacity region for the two-user MIMO IC with feedback $\mathbb{C}_{FB}$ and the capacity region for its corresponding reciprocal IC with feedback, $\mathbb{C}_{FB}^R$, are within constant gaps from each other. More precisely, the following expressions holds:
 \begin{eqnarray}
 \mathcal{R}_o(0) \ominus ([0,N_1+N_2]\times [0,N_1+N_2]) & \subseteq \mathbb{C}_{FB}\subseteq & \mathcal{R}_o(0)\oplus ([0,N_1]\times [0,N_2]),\label{rr1}\\
 \mathcal{R}_o(0) \ominus ([0,M_1+M_2]\times [0,M_1+M_2]) & \subseteq \mathbb{C}_{FB}^R\subseteq & \mathcal{R}_o(0)\oplus ([0,M_1]\times [0,M_2])\label{rr2}.
 \end{eqnarray}
 Then, we get
 \begin{eqnarray}
 \mathbb{C}_{FB}^R \ominus ([0,N_1+N_2+M_1]\times [0,N_1+N_2+M_2]) & \subseteq \mathbb{C}_{FB}\subseteq &\nonumber\\
 \mathbb{C}_{FB}^R\oplus ([0,M_1+M_2+N_1]\times [0,M_1+M_2+N_2]),&&\label{rr3}\\
 \mathbb{C}_{FB} \ominus ([0,M_1+M_2+N_1]\times [0,M_1+M_2+N_2]) & \subseteq \mathbb{C}_{FB}^R\subseteq &\nonumber\\
 \mathbb{C}_{FB}\oplus ([0,N_1+N_2+M_1]\times [0,N_1+N_2+M_2]).\label{rr4}&&
 \end{eqnarray}\label{thm:reciprocity}
 \end{theorem}

\begin{proof}
In Appendix \ref{Appendix2}, we show that the region $\mathcal{R}_o(0)$ for the MIMO IC is the same as the corresponding region $\mathcal{R}_o^R(0)$ for the corresponding reciprocal MIMO IC. Thus, (\ref{rr1})-(\ref{rr2}) follow from Theorem \ref{outer_inner_capacity_reciprocal}. Moreover, (\ref{rr3})-(\ref{rr4}) follow from simple manipulations on (\ref{rr1})-(\ref{rr2}).
\end{proof}
\begin{figure}[ht]
\centering{
	\includegraphics[width=8.5cm]{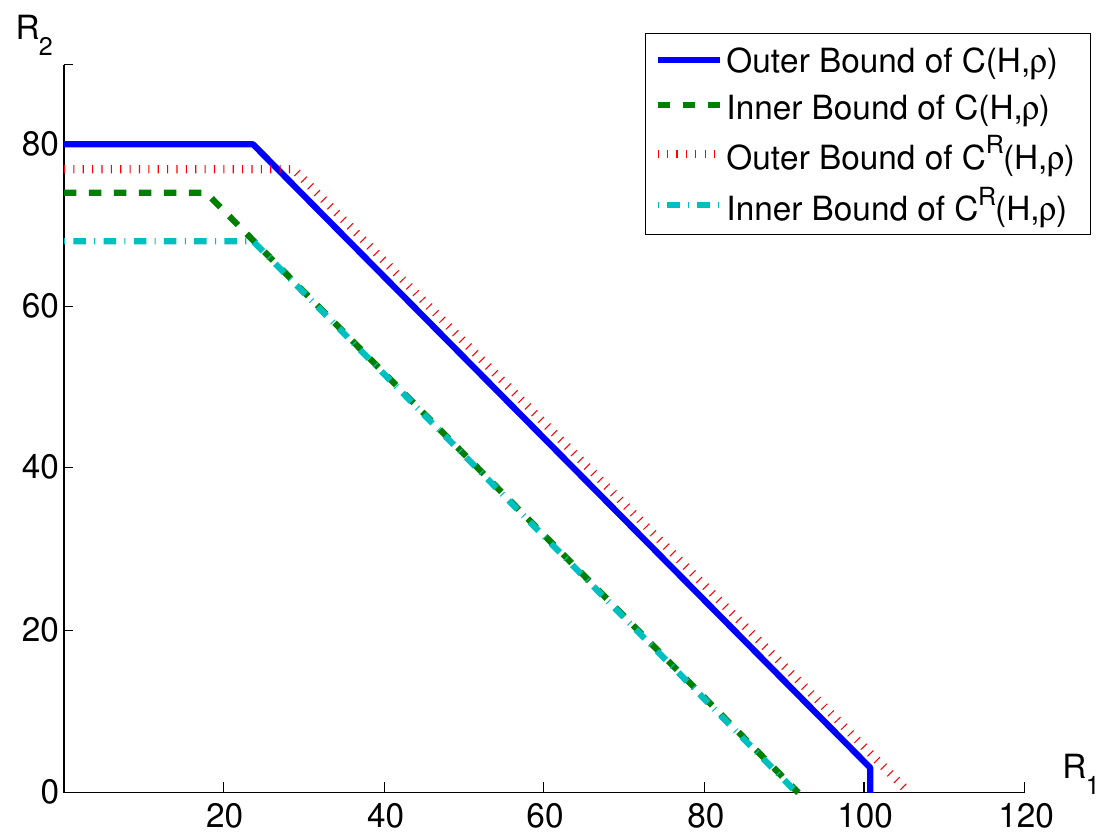}

}

\caption[Optional caption for list of figures]{Inner and outer bounds for the capacity region of MIMO IC with feedback specified in \eqref{hhhhhh} and inner and outer bounds for its reciprocal channel.
}
\label{fig:uibtt}
\end{figure}

Thus, we see that the capacity region of a two-user MIMO IC with feedback and the  corresponding reciprocal channel with feedback are within $N_1+N_2+M_1+M_2+\max{(N_1+M_1,N_2+M_2)}$ bits.

In Figure \ref{fig:uibtt},  we compare the inner and outer bounds for the capacity region of the MIMO IC with feedback specified in \eqref{hhhhhh}, and inner and outer bounds for its reciprocal channel. For this figure, the parameters for the IC are the same as those used for Figure \ref{fig:uibmm}. We note that the capacity region of the MIMO IC with feedback and that of its reciprocal channel with feedback are within a constant gap.





\section{GDoF Region of MIMO IC with Feedback}

This section describes our results on the GDoF region of the two-user MIMO IC with feedback. The GDoF gives the high SNR characterization of the capacity region. Since the inner and outer-bounds on the capacity region are within a constant gap, we characterize the exact GDoF region of the MIMO IC with feedback.

Define
\begin{eqnarray}\label{f}
f(u,\left({a}_{1},u_1\right),\left({a}_{2},u_2\right))\triangleq
\begin{cases}
\min{(u,u_1)}a^+_1+\min{((u-u_1)^+,u_2)}a_2^+, & \text{ if } a_1\ge a_2\\
\min{(u,u_2)}a^+_2+\min{((u-u_2)^+,u_1)}a_1^+, & \text{ otherwise }
\end{cases}.
\end{eqnarray}

The following result characterizes the GDoF for general MIMO IC with feedback for general power scaling parameters $\alpha_{ij}$.

\begin{theorem}
The GDoF region of the two-user MIMO IC with feedback is given by the set of $(d_1,d_2)$ satisfying:
\begin{eqnarray}
{\alpha }_{11}d_1&\le& f(N_1,\left({\alpha }_{11},M_1\right),\left({\alpha }_{21},M_2\right)),\label{gdofeq1}\\
{\alpha }_{22}d_2&\le& f(N_2,\left({\alpha }_{22},M_2\right),\left({\alpha }_{12},M_1\right)),\label{gdofeq2}\\
{\alpha }_{11}d_1&\le& {\alpha }_{12}{\min  \left(M_1,N_2\right)\ }+{\alpha }_{11}{\min  \left({\left(M_1-N_2\right)}^+,N_1\right)\ }+\nonumber\\
&&{\left({\alpha }_{11}-{\alpha }_{12}\right)}^+({\min  \left(M_1,N_1\right)\ }-{\min  \left({\left(M_1-N_2\right)}^+,N_1\right)\ }),\label{gdofeq3}\\
{\alpha }_{22}d_2&\le& {\alpha }_{21}{\min  \left(M_2,N_1\right)\ }+{\alpha }_{22}{\min  \left({\left(M_2-N_1\right)}^+,N_2\right)\ }+\nonumber\\
&&{\left({\alpha }_{22}-{\alpha }_{21}\right)}^+\left({\min  \left(M_2,N_2\right)\ }-{\min  \left({\left(M_2-N_1\right)}^+,N_2\right)\ }\right),\label{gdofeq4}\\
{\alpha }_{11}d_1+{\alpha }_{22}d_2&\le& f\left(N_2,\left({\alpha }_{22},M_2\right),\left({\alpha }_{12},M_1\right)\right)+{\alpha }_{11}{\min  \left({\left(M_1-N_2\right)}^+,N_1\right)\ }+\nonumber\\
&&{\left({\alpha }_{11}-{\alpha }_{12}\right)}^+\left({\min  \left(M_1,N_1\right)\ }-{\min  \left({\left(M_1-N_2\right)}^+,N_1\right)\ }\right),\label{gdofeq5}\\
{\alpha }_{11}d_1+{\alpha }_{22}d_2&\le& f\left(N_1,\left({\alpha }_{11},M_1\right),\left({\alpha }_{21},M_2\right)\right)+{\alpha }_{22}{\min  \left({\left(M_2-N_1\right)}^+,N_2\right)\ }+\nonumber\\
&&{\left({\alpha }_{22}-{\alpha }_{21}\right)}^+\left({\min  \left(M_2,N_2\right)\ }-{\min  \left({\left(M_2-N_1\right)}^+,N_2\right)\ }\right). \label{gdofeq6}
\end{eqnarray}

\label{thm_gdof}
\end{theorem}
\begin{proof}
According to Theorem \ref{outer_inner_capacity_reciprocal}, we can see that $GDoF=\lim_{\mathsf{SNR}\to \infty}\mathcal{R}_o(0)/\log \mathsf{SNR}$, which is evaluated in Appendix \ref{Appendix3} to get the result as in the statement of the theorem.
\end{proof}
Since the capacity region of the MIMO IC with feedback and the corresponding reciprocal IC with feedback are within constant gap, the GDoF region of the MIMO IC with feedback and that of the corresponding reciprocal IC with feedback are the same, as given in the next corollary.
\begin{corollary}
The GDoF region for the reciprocal IC with perfect feedback is given by the set of $(d_1,d_2)$ satisfying \eqref{gdofeq1}-\eqref{gdofeq6}.\label{cor:reciprocity}
\end{corollary}

We will now consider a special case of Theorem \ref{thm_gdof} where $M_1=M_2=M$, $N_1=N_2=N$, $\alpha_{11}=\alpha_{22}=1$, and $\alpha_{12}=\alpha_{21}=\alpha$. This MIMO IC is called a symmetric MIMO IC. We also define GDoF, $d$, as the supremum over all $d_i$ such that $(d_i,d_i)$ is in the GDoF region. The GDoF for the symmetric MIMO IC with feedback is given as follows.

\begin{corollary}
The GDoF for a two-user symmetric MIMO IC with feedback for $N\le M$ is given as follows:
\begin{eqnarray}
GDoF_{PF}=\left\{ \begin{array}{ll}
N-\frac{\alpha}{2} {\left(2N-M\right)}^+,& \text{ if }\alpha \le 1, \\
N(\frac{\alpha+1}{2})-\frac{1}{2}{\left(2N-M\right)}^+,& \text{ if } \alpha \ge 1. \end{array}
\right.
\end{eqnarray}
Since the expressions are symmetric in $N$ and $M$ by Corollary \ref{cor:reciprocity}, the GDoF for $M\le N$ follows by interchanging the roles of $M$ and $N$.
\end{corollary}
\begin{proof}
For the symmetric MIMO IC, we have
\begin{eqnarray}
f(N_i,\left({\alpha }_{ii},M_i\right),\left({\alpha }_{ji},M_j\right))&=&f(N,\left(1,M\right),\left({\alpha },M\right))\nonumber\\
&=& \max(1,\alpha)\min(M,N)+\min(1,\alpha)\min((N-M)^+,M).\label{fforsym}
\end{eqnarray}

We will split the proof for $N\le M$ in two cases.

\noindent {\bf Case 1 - $\alpha\le1$:}
We will go over all equations \eqref{gdofeq1}-\eqref{gdofeq6} and evaluate them for the symmetric case with $\alpha\le 1$.
Equations (\ref{gdofeq1}) and (\ref{gdofeq2}) can be simplified using \eqref{fforsym} as follows
\begin{eqnarray}
d&\le& \max(1,\alpha)\min(M,N)+\min(1,\alpha)\min((N-M)^+,M)\nonumber\\
&=&N. \label{simp1l1}
\end{eqnarray}
Equations (\ref{gdofeq3}) and (\ref{gdofeq4}) can be simplified as
\begin{eqnarray}
d&\le& {\alpha }{\min  \left(M,N\right)\ }+{\min  \left({\left(M-N\right)}^+,N\right)\ }+{\left(1-{\alpha }\right)}^+({\min  \left(M,N\right)\ }-{\min  \left({\left(M-N\right)}^+,N\right)\ })\nonumber\\
&=& \alpha N+\min((M-N),N)+{(1-{\alpha })}{N}-{(1-\alpha)}{\min  \left({\left(M-N\right)},N\right)\ }\nonumber\\
&=&N+\alpha\min((M-N),N)\nonumber\\
&=&N+\alpha(N-(2N-M)^+). \label{simp2l1}
\end{eqnarray}
Equations (\ref{gdofeq5}) and (\ref{gdofeq6}) can be simplified as
\begin{eqnarray}
d&\le& \frac{1}{2}(\max(1,\alpha)\min(M,N)+\min(1,\alpha)\min((N-M)^+,M)+{\min  \left({\left(M-N\right)}^+,N\right)\ }+\nonumber\\
&&{\left(1-{\alpha }\right)}^+({\min  \left(M,N\right)\ }-{\min  \left({\left(M-N\right)}^+,N\right)\ }))\nonumber\\
&=&\frac{1}{2}(N+(1-\alpha)N+\alpha\min((M-N),N))\nonumber\\
&=&N-\frac{1}{2}\alpha(N-(N-(2N-M)^+))\nonumber\\
&=&N-\frac{\alpha}{2}((2N-M)^+).  \label{simp3l1}
\end{eqnarray}

We note that the minimum of the right hand sides of \eqref{simp1l1}, \eqref{simp2l1}, and \eqref{simp3l1} would give us the GDoF. The minimum of these three terms is \eqref{simp3l1} which proves the result for $\alpha \le 1$.

\noindent {\bf Case 2 - $\alpha\ge1$:} In this case, equations (\ref{gdofeq1}) and (\ref{gdofeq2}) can be simplified as
\begin{eqnarray}
d&\le& \max(1,\alpha)\min(M,N)+\min(1,\alpha)\min((N-M)^+,M)\nonumber\\
&=&\alpha N. \label{simp1g1}
\end{eqnarray}
Equations (\ref{gdofeq3}) and (\ref{gdofeq4}) can be simplified as
\begin{eqnarray}
d&\le& {\alpha }{\min  \left(M,N\right)\ }+{\min  \left({\left(M-N\right)}^+,N\right)\ }+\nonumber\\
&&{\left(1-{\alpha }\right)}^+({\min  \left(M,N\right)\ }-{\min  \left({\left(M-N\right)}^+,N\right)\ })\nonumber\\
&=& \alpha N+\min((M-N),N). \label{simp2g1}
\end{eqnarray}
Equations (\ref{gdofeq5}) and (\ref{gdofeq6}) can be simplified as
\begin{eqnarray}
d&\le& \frac{1}{2}(\max(1,\alpha)\min(M,N)+\min(1,\alpha)\min((N-M)^+,M)+{\min  \left({\left(M-N\right)}^+,N\right)\ }+\nonumber\\
&&{\left(1-{\alpha }\right)}^+({\min  \left(M,N\right)\ }-{\min  \left({\left(M-N\right)}^+,N\right)\ }))\nonumber\\
&=&\frac{1}{2}(\alpha N+(N-(2N-M)^+))\nonumber\\
&=&N\frac{(\alpha+1)}{2}-\frac{1}{2}(2N-M)^+. \label{simp3g1}
\end{eqnarray}
We note that the minimum of the right hand sides of \eqref{simp1g1}, \eqref{simp2g1}, and \eqref{simp3g1} would give us the GDoF. The minimum of these three terms is \eqref{simp3g1} which proves the result for $\alpha \ge 1$.
\end{proof}

The authors of  \cite{gdof} found the GDoF for the two-user symmetric MIMO IC without feedback as follows for $N\le M$ (We can interchange the roles of $N$ and $M$ if $N>M$.)
\begin{eqnarray}
GDoF_{NF}=\left\{ \begin{array}{ll}
N-{\alpha}{\left(2N-M\right)}^+,&  \text{ if }0\le\alpha\le\frac{1}{2},\\
N-(1-{\alpha}){\left(2N-M\right)}^+,& \text{ if }\frac{1}{2}\le\alpha\le \frac{2}{3} ,\\
N-\frac{\alpha}{2}{\left(2N-M\right)}^+,& \text{ if }\frac{2}{3}\le\alpha\le1,\\
\min\{N,N(\frac{\alpha+1}{2})-\frac{1}{2}{\left(2N-M\right)}^+\},& \text{ if } 1 \le \alpha. \end{array}
\right.
\end{eqnarray}

We note that the GDoF with and without feedback are the same for $\frac{2}{3}\le\alpha \le 1$. Figure \ref{fig:subfigureExample} compares the GDoF for the two-user symmetric MIMO IC with and without feedback. In Figure \ref{fig:subfigureExample}(a), the ``W"-curve obtained without feedback delineates the very weak ($0\le\alpha\le\frac{1}{2}$), weak ($\frac{1}{2}\le\alpha\le\frac{2}{3}$), moderate ($\frac{2}{3}\le\alpha\le1$), strong ($1\le\alpha\le{3-{\frac{M}{N}}}$) and very strong ($3-{\frac{M}{N}}\le\alpha$) interference regimes. In the presence of feedback, the ``W"-curve improves to a ``V"-curve which delineates the weak ($0\le\alpha\le1$) and strong ($1\le\alpha$) interference regimes for all choices of $N$ and $M$. For $\frac{M}{2}<N\leq M$, we see that the GDoF with feedback is strictly greater than that without feedback for $0<\alpha< 2/3$ and for $\alpha > 3-M/N$. For $N\le M/2$,  we see that the GDoF with feedback is strictly greater than that without feedback for $\alpha > 2$. The GDoF improvement indicates an unbounded gap in the corresponding capacity regions as the SNR goes to infinity.

\begin{figure}[ht]
\centering
\subfigure[$\frac{M}{2}<N\leq M\quad\left(\alpha_0=3-\frac{M}{N}\Rightarrow 1<\alpha_0\leq2\right)$]{
	\includegraphics[width=8.5cm]{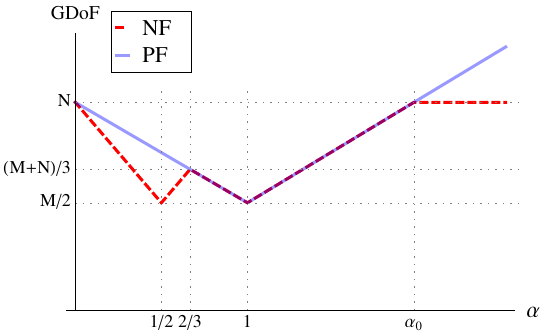}
    \label{fig:subfig1}
}
\subfigure[$N\leq \frac{M}{2}$]{
	\includegraphics[width=8.5cm]{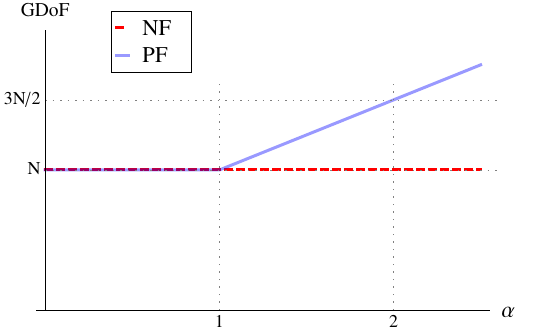}
    \label{fig:subfig2}
}
\caption[Optional caption for list of figures]{GDoF for symmetric MIMO IC with perfect feedback (PF), and no-feedback (NF) for \subref{fig:subfig1} $\frac{M}{2}<N\leq M$, and \subref{fig:subfig2} $N\leq \frac{M}{2}$.}
\label{fig:subfigureExample}
\end{figure}
Interestingly, from Figure \ref{fig:subfigureExample}(b) we can see that if we increase $M$ when $N\leq\frac{M}{2}$, the GDoF does not change. This can be interpreted as that while $N\leq\frac{M}{2}$, $N$ act as a bottleneck and increasing $M$ does not increase the GDoF. As a special case consider a MISO IC for which we note that the GDoF is the same for all $M\ge 2$. Thus, increasing the transmit antennas beyond $2$ does not increase the GDoF. However, increasing the transmit antennas from $1$ to $2$ gives a strict improvement in GDoF for all $\alpha>0$. Similar result also holds for SIMO systems where increasing the receive antennas from $1$ to $2$ help increase GDoF while increasing the receive antennas beyond $2$ does not increase the GDoF.




\section{Conclusions}
This paper gives the capacity region of the MIMO IC with feedback within $N_1+N_2+\max(N_1,N_2)$ bits. The achievability is based on the block Markov encoding, backward decoding, and Han-Kobayashi message-splitting. The capacity region for the MIMO IC with feedback is shown to be within a constant number of bits from the capacity region of the corresponding reciprocal IC. Further, the GDoF region for the general MIMO IC is characterized. It is found that for the symmetric IC with feedback, the GDoF form a ``V"-curve rather than the ``W"-curve without feedback.

{ The authors of \cite{Alireza} considered a SISO IC with two rate-limited feedback links. Further, the authors of \cite{Achal} considered nine canonical feedback models in the SISO IC, ranging from one feedback link to four feedback links in various configurations. Extension of this work for different feedback models proposed in \cite{Achal} for rate-limited feedback links is an important future work, and is still open. Further, the extension to the general $K$-user IC is also open.}




\begin{appendices}
\section{Proof of Outer Bound for Theorem \ref{outer_capacity}} \label{apdx_outer}

In this Appendix, we will show that $\mathbb{C}_{FB}\subseteq \mathcal{R}_o(Q)$ for some covariance matrix $Q=\mathbb{E}[X_{1}X_{2}^{\dagger }]$.

{ The set of upper bounds to the capacity region will be derived in two steps. First, the capacity region is outer-bounded by a region defined in terms of the differential entropy of the random variables associated with the signals. These outer-bounds use genie-aided information at the receivers. Second, we outer-bound this region to prove the outer-bound as described in the statement of Theorem \ref{outer_capacity}.}

The following result outer-bounds the capacity region of two-user MIMO IC with feedback.

\begin{lemma}
Let $S_i$ be defined as $S_i\triangleq \sqrt{{\rho }_{ij}}H_{ij}X_i+Z_j$. Then, the capacity region of a two-user MIMO IC with feedback is outerbounded by the region formed by $(R_1,R_2)$ satisfying
\begin{eqnarray}
R_1&\le& h(Y_1)-h(Z_1),\\
R_2&\le& h(Y_2)-h(Z_2),\\
R_1&\le& h\left(Y_2\mathrel{\left|\vphantom{Y_2 X_2}\right.\kern-\nulldelimiterspace}X_2\right)-h\left(Z_2\right)+h(Y_1|X_2,S_1)-h(Z_1),\\
R_2&\le& h\left(Y_1\mathrel{\left|\vphantom{Y_1 X_1}\right.\kern-\nulldelimiterspace}X_1\right)-h\left(Z_1\right)+h(Y_2|X_2,S_1)-h(Z_2),\\
R_1+R_2&\le& h\left(Y_1\mathrel{\left|\vphantom{Y_1 S_1,X_2}\right.\kern-\nulldelimiterspace}S_1,X_2\right)-h\left(Z_2\right)+h(Y_2)-h(Z_1),\\
R_1+R_2&\le& h\left(Y_2\mathrel{\left|\vphantom{Y_2 S_2,X_1}\right.\kern-\nulldelimiterspace}S_2,X_1\right)-h\left(Z_1\right)+h(Y_1)-h(Z_2).
\end{eqnarray}
\label{lemma_conditional}
\end{lemma}

\begin{proof}
The proof follows the same lines as the proof of Theorem 3 in \cite{Tse}, replacing SISO channel gains by MIMO channel gains and is thus omitted here.
\end{proof}

{ The rest of the section outer-bounds this region to get the outer bound in Theorem \ref{outer_capacity}. For this, we will introduce some useful Lemmas.}

The next result  outer-bounds the entropies and the conditional entropies of two random variables by their corresponding Gaussian random variables.

\begin{lemma}[\cite{Gaussian}]
Let $X$ and $Y$ be two random vectors, and let $X^G$ and $Y^G$ be Gaussian vectors with covariance matrices satisfying
\begin{eqnarray}
Cov\left[ \begin{array}{c}
X \\
Y \end{array}
\right]=Cov\left[ \begin{array}{c}
X^G \\
Y^G \end{array}
\right],
\end{eqnarray}
Then, we have
\begin{eqnarray}
h(Y)&\le& h(Y^G),\\
h\left(Y\mathrel{\left|\vphantom{Y X}\right.\kern-\nulldelimiterspace}X\right)&\le& h\left(Y^G\mathrel{\left|\vphantom{Y^G X^G}\right.\kern-\nulldelimiterspace}X^G\right).\label{gauss_cond}
\end{eqnarray}
\label{gauss_ob}
\end{lemma}

The next result gives the determinant of a block matrix, which will be used extensively in the sequel.

\begin{lemma}[\cite{det}]
For block matrix $M=\left[ \begin{array}{cc}
A & B \\
C & D \end{array}
\right]$ with matrices A, B, C, and D, we have:
\begin{eqnarray}
\det  M= \begin{cases}\det  A \det (D-CA^{-1}B), & \text{if A is invertible,}\label{eq1_block}\\
\det D \det (A-BD^{-1}C), & \text{if D is invertible.} \label{eq2_block}
\end{cases}
\end{eqnarray}\label{lem_block}
\end{lemma}
Now, we introduce a lemma that is a key  result which will be used to upper-bound a conditional entropy term in this section and also to show an upper bound in Appendix \ref{Appendix4}.

\begin{lemma} \label{L}
Let $L(K,S)$ be defined as
\begin{eqnarray}
L\left(K, S\right)
\triangleq K-KS{(I_{N_2}+S^{\dagger }KS)}^{-1}S^{\dagger }K, \label{eqdefnl}
\end{eqnarray}
for some $M_1\times M_1$ p.s.d. Hermitian matrix $K$ and some $M_1\times N_2$ matrix $S$.
Then if $0 \preceq K_1\preceq K_2$ for some Hermitian matrices $K_1$ and $K_2$, we have
\begin{eqnarray}\label{tpl}
L\left(K_1, S\right)\preceq L\left(K_2, S\right)\label{ll}.
\end{eqnarray}
\end{lemma}
\begin{proof}
We note that since $K$ is p.s.d., $K+\epsilon I_{M_1}$ is invertible for all $\epsilon>0$. Given $0\preceq K_1\preceq K_2$, let $F(\epsilon)\triangleq L(K_2+\epsilon I_{M_1},S) - L(K_1+\epsilon I_{M_1},S)$. We need to show that $F(0)\succeq 0$.

We first show that $F(\epsilon)\succeq 0$ for all $\epsilon>0$. From Woodbury matrix identity (Appendix C.4.3 of \cite{Boyd}), we have that if $A$ is invertible, $(A+BD)^{-1}= A^{-1}-A^{-1}B(I+DA^{-1}B)^{-1}DA^{-1}$. Thus,  we have $L(K+\epsilon I_{M_1},S)  = ((K+\epsilon I_{M_1})^{-1}+SS^\dagger)^{-1}$ by substituting $A$ as $(K+\epsilon I_{M_1})^{-1}$, $B$ as $S$ and $D$ as $S^\dagger$ in the above identity.

Thus, $F(\epsilon) = ((K_2+\epsilon I_{M_1})^{-1}+SS^\dagger)^{-1} - ((K_1+\epsilon I_{M_1})^{-1}+SS^\dagger)^{-1}$. Since $K_1$ and $K_2$ are Hermitian p.s.d. matrices with $K_1 \preceq K_2$, it easily follows that $F(\epsilon)\succeq 0$.

Having shown that $F(\epsilon)\succeq 0$ for all $\epsilon>0$, we will now prove the continuity of $F(\epsilon)$ at $\epsilon =0$. For this, we take the partial derivative of  $F(\epsilon)$ at $\epsilon = 0$ and show that it is not unbounded thus proving that $F(\epsilon)$ is continuous at $\epsilon = 0$. Thus, we have
\begin{eqnarray}
\frac{d F(\epsilon)}{d \epsilon}&=& \frac{d}{d \epsilon} L(K_2+\epsilon I_{M_1},S) - L(K_1+\epsilon I_{M_1},S)\nonumber\\
&=& \frac{d}{d \epsilon} L(K_2+\epsilon I_{M_1},S) - \frac{d}{d \epsilon}L(K_1+\epsilon I_{M_1},S).
\end{eqnarray}
Thus, it is enough to show that $\lim_{\epsilon \to 0} \frac{d}{d \epsilon}L(K_i+\epsilon I_{M_1},S)$ is bounded. We have
\begin{eqnarray}
\lim_{\epsilon \to 0} \frac{d}{d \epsilon}L(K_i+\epsilon I_{M_1},S)&=& \lim_{\epsilon \to 0} \frac{d}{d \epsilon}( K_i+\epsilon I_{M_1}-(K_i+\epsilon I_{M_1})S{(I_{N_2}+S^{\dagger }(K_i+\epsilon I_{M_1})S)}^{-1}S^{\dagger }(K_i+\epsilon I_{M_1}))\nonumber\\
&=& I_{M_1} - \lim_{\epsilon \to 0} \frac{d}{d \epsilon}((K_i+\epsilon I_{M_1})S{(I_{N_2}+S^{\dagger }(K_i+\epsilon I_{M_1})S)}^{-1}S^{\dagger }(K_i+\epsilon I_{M_1}))\nonumber\\
&=& I_{M_1} - S{(I_{N_2}+S^{\dagger }K_i S)}^{-1}S^{\dagger }K_i- K_i S{(I_{N_2}+S^{\dagger }K_i S)}^{-1}S^{\dagger }\nonumber\\&&+ K_i S{(I_{N_2}+S^{\dagger }K_i S)}^{-1} S^\dagger S {(I_{N_2}+S^{\dagger }K_i S)}^{-1}S^{\dagger }K_i,
\end{eqnarray}
which is bounded. Hence, $F(\epsilon)$ is continuous at $\epsilon = 0$. Further, since $K_1$ and $K_2$ are Hermitian, we see that $F(\epsilon)$ is Hermitian and thus normal. From the Wielandt-Hoffman theorem \cite{K}, we note that the $\mathbb{L}_2$ norm of the difference in eigen-values (ordered in a particular way) of two normal matrices is bounded by the Frobenium norm of the difference of the two matrices. This shows that since $F(\epsilon)\succeq 0$ and $F(\epsilon)-F(0) \to 0$ as $\epsilon \to 0$, we have that the eigen-values of $F(\epsilon)$ approach the eigen-values of $F(0)$ as $\epsilon \to 0$. Therefore, all the eigen-values of $F(0)$ are non-negative which proves that $F(0)$ is positive semi-definite thus proving the result.
\end{proof}

The next three Lemmas outer-bounds entropy and conditional entropies of some random variables.

\begin{lemma}
The entropy of the received signal at the $i^{\text{th}}$ receiver, $h(Y_i)$, is outer-bounded as follows
\begin{eqnarray}
h(Y_i)&\le& {\log  {\det  \left(I_{N_i}+{\rho }_{ii}H_{ii}H^{\dagger }_{ii}+{\rho }_{ji}H_{ji}H^{\dagger }_{ji}+\sqrt{{\rho }_{ii}{\rho }_{ji}}H_{ii}Q_{ij}H^{\dagger }_{ji}+\sqrt{{\rho }_{ii}{\rho }_{ji}}H_{ji}Q_{ij}^{\dagger }H^{\dagger }_{ii}\right)\ }\ }\nonumber\\
&&+{N_i}\log \left( \pi e\right),
\end{eqnarray}
for $i,j\in \{1,2\}$, $i\ne j$. \label{lemma_hyi}
\end{lemma}

\begin{proof}
\begin{eqnarray}
h(Y_i) &\stackrel{(a)}{\le}& h(Y^G_i)\nonumber\\
&=&\log  \det  \pi e\left(I_{N_i}+{\rho }_{ii}H_{ii}Q_{ii}H^{\dagger }_{ii}+{\rho }_{ji}H_{ji}Q_{jj}H^{\dagger }_{ji}+\sqrt{{\rho }_{ii}{\rho }_{ji}}H_{ii}Q_{ij}H^{\dagger }_{ji}\right.\nonumber\\
 & & \left. +  \sqrt{{\rho }_{ii}{\rho }_{ji}}H_{ji}Q_{ij}^{\dagger }H^{\dagger }_{ii}\right)\nonumber\\
&\stackrel{(b)}{\le}& {\log  {\det  \left(I_{N_i}+{\rho }_{ii}H_{ii}H^{\dagger }_{ii}+{\rho }_{ji}H_{ji}H^{\dagger }_{ji}+\sqrt{{\rho }_{ii}{\rho }_{ji}}H_{ii}Q_{ij}H^{\dagger }_{ji}+\sqrt{{\rho }_{ii}{\rho }_{ji}}H_{ji}Q_{ij}^{\dagger }H^{\dagger }_{ii}\right)\ }\ }\nonumber\\
&&+{N_i}\log \left( \pi e\right),
\end{eqnarray}
where (a) follows from Lemma \ref{gauss_ob}, and (b) follows from the fact that $\log \det (.)$ is a monotonically increasing function on the cone of positive definite matrices and we have $Q_{ii}\preceq I_{M_i}$ for $i\in \left\{1,2\right\}$.

Taking $\pi e$ out of the above determinant in the last part, gives the result as in the statement of the Lemma.
\end{proof}

\begin{lemma}\label{yixi}
The conditional entropy of the received signal at the $i^{\text{th}}$ receiver given the transmitted signal from the $i^{\text{th}}$ transmitter, $h(Y_i|X_i)$ is outer-bounded as follows
\begin{equation}
h\left(Y_i|X_i\right)\le {\log  {\det  \left(I_{N_i}+{\rho }_{ji}H_{ji}H^{\dagger }_{ji}-{\rho }_{ji}H_{ji}Q_{ij}^{\dagger }Q_{ij}H^{\dagger }_{ji}\right)}}+{N_i}\log \left( \pi e\right),
\end{equation}
where $Q_{ij}$ is the cross-covariance between $X_i$ and $X_j$ and $Q_{ii}$ is the covariance matrix for $X_i$.
\end{lemma}
\begin{proof}
Let
\begin{eqnarray}\label{K}
&&K_{i1}\triangleq \mathbb{E}\left[ \begin{array}{cc}
X_iX^{\dagger }_i & X_iY^{\dagger }_i \\
Y_iX^{\dagger }_i & Y_iY^{\dagger }_i \end{array}
\right]=\\
&&\left[ \begin{array}{cc}
Q_{ii} & {\sqrt{{\rho }_{ii}}Q}_{ii}H^{\dagger }_{ii}+\sqrt{{\rho }_{ji}}Q_{ij}H^{\dagger }_{ji}\nonumber\\
\sqrt{{\rho }_{ii}}H_{ii}Q_{ii}+\sqrt{{\rho }_{ji}}H_{ji}Q_{ij}^{\dagger } & \mathbb{E}[Y_iY^{\dagger }_i] \end{array}
\right].
\end{eqnarray}

where
\begin{eqnarray}
\mathbb{E}[Y_iY^{\dagger }_i]=
I_{N_i}+{\rho }_{ii}H_{ii}Q_{ii}H^{\dagger }_{ii}+{\rho }_{ji}H_{ji}Q_{jj}H^{\dagger }_{ji}+\sqrt{{\rho }_{ii}{\rho }_{ji}}H_{ii}Q_{ij}H^{\dagger }_{ji}+\sqrt{{\rho }_{ii}{\rho }_{ji}}H_{ji}Q_{ij}^{\dagger }H^{\dagger }_{ii},
\end{eqnarray}

and
\begin{eqnarray}
K_{i2}\triangleq \mathbb{E}[X_iX^{\dagger }_i]=Q_{ii}.
\end{eqnarray}
According to Lemma \ref{gauss_ob}, we get
\begin{eqnarray}
h(Y_i|X_i)&\le& h(Y^G_i|X^G_i)\nonumber\\
&=& h(X^G_i,Y^G_i)-h(X^G_i)\nonumber\\
&=& \log \det \pi e(K_{i1})-\log \det \pi e(K_{i2})\nonumber\\
&=& \log \det (K_{i1})-\log \det (K_{i2})+\log \det \pi e(I_{N_i}). \label{eq_sumyixi}
\end{eqnarray}

Due to the reason that $Q$'s elements are chosen from a continuous space, it is invertible with probability of one. In addition, according to Corollary $7.7.4(a)$ of \cite{Q}, if we have $Q_{ii}\preceq I_{M_i}$, $Q^{-1}_{ii}\succeq I_{M_i}$.
Using Lemma \ref{lem_block} with $M=K_{i1}$ and $A=K_{i2}$, we get
\begin{eqnarray}
\log \det K_{i1} &=& \log \left({\det \mathbb{E}(X_iX^{\dagger }_i) \det (\mathbb{E}(Y_iY^{\dagger }_i)-\mathbb{E}(Y_iX^{\dagger }_i)(\mathbb{E}(X_iX^{\dagger }_i))^{-1}\mathbb{E}(X_iY^{\dagger }_i))}\right)\nonumber\\
&=& \log {\det \left(\mathbb{E}(X_iX^{\dagger }_i)\right)+ \log \det \left(\mathbb{E}(Y_iY^{\dagger }_i)-\mathbb{E}(Y_iX^{\dagger }_i)(\mathbb{E}(X_iX^{\dagger }_i))^{-1}\mathbb{E}(X_iY^{\dagger }_i)\right)}\nonumber\\
&\stackrel{(a)}{=}& \log \det \left(Q_{ii}\right)+\log  \det \left(I_{N_i}+{\rho }_{ji}H_{ji}Q_{jj}H^{\dagger }_{ji}-{\rho }_{ji}H_{ji}Q_{ij}^{\dagger }Q^{-1}_{ii}Q_{ij}H^{\dagger }_{ji}\right)\nonumber\\
&\stackrel{(b)}{\le}& \log \det \left(Q_{ii}\right) + \log \det \left(I_{N_i}+{\rho }_{ji}H_{ji}H^{\dagger }_{ji}-{\rho }_{ji}H_{ji}Q_{ij}^{\dagger }Q_{ij}H^{\dagger }_{ji}\right),\label{eq_k1}
\end{eqnarray}
where $(a)$ is obtained by using (\ref{K}) and some simplifications, and $(b)$ follows from the fact that log det (.) is a monotonically increasing function on the cone of positive definite matrices and we have $Q_{ii}\preceq I_{M_i}$ and $Q^{-1}_{ii}\succeq I_{M_i}$ according to Corollary 7.7.4(a) of \cite{Q} for $i,j\in \left\{1,2\right\},\ i\ne j$.


Substituting \eqref{eq_k1} in \eqref{eq_sumyixi} gives the result as in the statement of the lemma.
\end{proof}

\begin{lemma}
The conditional entropy of the received signal at the $i^{\text{th}}$ receiver given $X_j$ and $S_i$, $h(Y_i|X_j,S_i)$ is outer-bounded as follows
\begin{eqnarray}
h\left(Y_i\mathrel{\left|\vphantom{Y_i X_j,S_i}\right.\kern-\nulldelimiterspace}X_j,S_i\right)&\le&
 \log  \det  \Bigg(I_{N_i}+{\rho }_{ii}H_{ii}H^{\dagger }_{ii}-\left[ \begin{array}{cc}
\sqrt{{\rho }_{ii}{\rho }_{ij}}H_{ii}H^{\dagger }_{ij} & \sqrt{{\rho }_{ii}}H_{ii}Q_{ij} \end{array}
\right]\nonumber\\
&&{\left[ \begin{array}{cc}
I_{N_j}+{\rho }_{ij}H_{ij}H^{\dagger }_{ij} & \sqrt{{\rho }_{ij}}H_{ij}Q_{ij} \\
\sqrt{{\rho }_{ij}}Q_{ij}^{\dagger }H^{\dagger }_{ij} & I_{M_j} \end{array}
\right]}^{-1}\left[ \begin{array}{c}
\sqrt{{\rho }_{ii}{\rho }_{ij}}H_{ij}H^{\dagger }_{ii} \\
\sqrt{{\rho }_{ii}}Q_{ij}^{\dagger }H^{\dagger }_{ii} \end{array}
\right]\Bigg)\nonumber\\
&&+{N_i}\log \left( \pi e\right).
\end{eqnarray}\label{yixisi}
\end{lemma}
\begin{proof}
Let $K_{i3}$ and $K_{i4}$ be defined as follows
\begin{eqnarray}
K_{i3}&\triangleq& \mathbb{E}\left[\left[ \begin{array}{c}
\sqrt{{\rho }_{ii}}H_{ii}X_i+Z_i \\
\sqrt{{\rho }_{ij}}H_{ij}X_i+Z_j \\
X_j \end{array}
\right].{\left[ \begin{array}{c}
\sqrt{{\rho }_{ii}}H_{ii}X_i+Z_i \\
\sqrt{{\rho }_{ij}}H_{ij}X_i+Z_j \\
X_j \end{array}
\right]}^{\dagger }\right]\nonumber\\
&=&\left[ \begin{array}{ccc}
I_{N_i}+{\rho }_{ii}H_{ii}Q_{ii}H^{\dagger }_{ii} & \sqrt{{\rho }_{ii}{\rho }_{ij}}H_{ii}Q_{ii}H^{\dagger }_{ij} & \sqrt{{\rho }_{ii}}H_{ii}Q_{ij} \\
\sqrt{{\rho}_{ii}{\rho }_{ij}}H_{ij}Q_{ii}H^{\dagger }_{ii} & I_{N_j}+{\rho }_{ij}H_{ij}Q_{ii}H^{\dagger }_{ij} & \sqrt{{\rho }_{ij}}H_{ij}Q_{ij} \\
\sqrt{{\rho }_{ii}}Q_{ij}^{\dagger }H^{\dagger }_{ii} & \sqrt{{\rho }_{ij}}Q_{ij}^{\dagger }H^{\dagger }_{ij} & Q_{jj} \end{array}
\right],
\end{eqnarray}
and
\begin{eqnarray}
K_{i4}&\triangleq& \mathbb{E}\left[\left[ \begin{array}{c}
\sqrt{{\rho }_{ij}}H_{ij}X_i+Z_j \\
X_j \end{array}
\right]{\left[ \begin{array}{c}
\sqrt{{\rho }_{12}}H_{ij}X_i+Z_j \\
X_j \end{array}
\right]}^{\dagger }\right]\nonumber\\
&=&\left[ \begin{array}{cc}
I_{N_j}+{\rho }_{ij}H_{ij}Q_{ii}H^{\dagger }_{ij} & \sqrt{{\rho }_{ij}}H_{ij}Q_{ij} \\
\sqrt{{\rho }_{ij}}Q_{ij}^{\dagger }H^{\dagger }_{ij} & Q_{jj} \end{array}
\right].
\end{eqnarray}

Further, let $Y_i' =\sqrt{{\rho }_{ii}}H_{ii}X_i+Z_i$. Then,
\begin{eqnarray}
h(Y_i|X_j,S_i)&=&h\left(\sqrt{{\rho }_{ii}}H_{ii}X_i+\sqrt{{\rho }_{ji}}H_{ji}X_j+Z_i\mathrel{\left|\vphantom{\sqrt{{\rho }_{ii}}H_{ii}X_i+\sqrt{{\rho }_{ji}}H_{ji}X_j+Z_i X_j,\sqrt{{\rho }_{ij}}H_{ij}X_i+Z_j}\right.\kern-\nulldelimiterspace}X_j,\sqrt{{\rho }_{ij}}H_{ij}X_i+Z_j\right)\nonumber\\
&=&h\left(\sqrt{{\rho }_{ii}}H_{ii}X_i+Z_i\mathrel{\left|\vphantom{\sqrt{{\rho }_{ii}}H_{ii}X_i+Z_i X_j,\sqrt{{\rho }_{ij}}H_{ij}X_i+Z_j}\right.\kern-\nulldelimiterspace}X_j,\sqrt{{\rho }_{ij}}H_{ij}X_i+Z_j\right)\nonumber\\
&=& h(Y_i'|X_j,S_i)\nonumber\\
&\stackrel{(a)}{\le}& h(Y'^G_i|S^G_i,X^G_j)\nonumber\\
&=&h(Y'^G_i,S^G_i,X^G_j)-h(S^G_i,X^G_j)\nonumber\\
&=&\log \det \pi e(K_{i3})-\log \det \pi e(K_{i4})\nonumber\\
&=&\log \det (K_{i3})-\log \det (K_{i4})+{N_i}\log \left( \pi e\right),\label{yixjsi}
\end{eqnarray}
where (a) follows from Lemma \ref{gauss_ob} by taking the two vectors $S_i$ and $X_j$ of lengths $N_j$ and $M_j$, respectively, together as a single vector of length of $N_j+M_j$ and then, used Lemma \ref{gauss_ob}.

Substituting $M=K_{i3}$ and $D=K_{i4}$ in Lemma \ref{lem_block}, we get
\begin{eqnarray}
{\log  {\det  \left(K_{i3}\right)\ }\ }&=&{\log  {\det  \left(K_{i4}\right)\ }\ } + \log  \det  \Bigg(I_{N_i}+{\rho }_{ii}H_{ii}Q_{ii}H^{\dagger }_{ii}\nonumber\\
&&-\left[ \begin{array}{cc}
\sqrt{{\rho }_{ii}{\rho }_{ij}}H_{ii}Q_{ii}H^{\dagger }_{ij} & \sqrt{{\rho }_{ii}}H_{ii}Q_{ij} \end{array}
\right]{\left[ \left(K_{i4}\right)\
\right]}^{-1}\left[ \begin{array}{c}
\sqrt{{\rho }_{ii}{\rho }_{ij}}H_{ij}Q_{ii}H^{\dagger }_{ii} \\
\sqrt{{\rho }_{ii}}Q_{ij}^{\dagger }H^{\dagger }_{ii} \end{array}
\right]\Bigg)\nonumber\\
&=&{\log  {\det  \left(K_{i4}\right)\ }\ }+\log  \det  \Bigg(I_{N_i}+{\rho }_{ii}H_{ii}Q_{ii}H^{\dagger }_{ii}-\left[ \begin{array}{cc}
\sqrt{{\rho }_{ii}{\rho }_{ij}}H_{ii}Q_{ii}H^{\dagger }_{ij} & \sqrt{{\rho }_{ii}}H_{ii}Q_{ij} \end{array}
\right]\nonumber\\
&&{\left[ \begin{array}{cc}
I_{N_j}+{\rho }_{ij}H_{ij}Q_{ii}H^{\dagger }_{ij} & \sqrt{{\rho }_{ij}}H_{ij}Q_{ij} \\
\sqrt{{\rho }_{ij}}Q_{ij}^{\dagger }H^{\dagger }_{ij} & Q_{jj} \end{array}
\right]}^{-1}\left[ \begin{array}{c}
\sqrt{{\rho }_{ii}{\rho }_{ij}}H_{ij}Q_{ii}H^{\dagger }_{ii} \\
\sqrt{{\rho }_{ii}}Q_{ij}^{\dagger }H^{\dagger }_{ii} \end{array}\label{ki3}
\right]\Bigg).
\end{eqnarray}

Note that since $Q_{jj}\preceq I_{M_{j}}$, using Lemma \ref{lem_block} we can see that $Q_{jj}=I_{M_{j}}$ outer-bounds the determinant of
\begin{eqnarray}
{\left[ \begin{array}{cc}
I_{N_j}+{\rho }_{ij}H_{ij}Q_{ii}H^{\dagger }_{ij} & \sqrt{{\rho }_{ij}}H_{ij}Q_{ij} \\
\sqrt{{\rho }_{ij}}Q_{ij}^{\dagger }H^{\dagger }_{ij} & Q_{jj} \end{array}
\right]}.\nonumber
\end{eqnarray}

Since $B \preceq I_{M_{j}}$ implies $ABA^\dagger \preceq AA^\dagger$, we have that $Q_{jj}=I_{M_{j}}$ outer-bounds the expression of the right hand side of \eqref{ki3}. Thus,
\begin{eqnarray}
&&{\log  {\det  \left(K_{i3}\right)\ }\ }\le {\log  {\det  \left(K_{i4}\right)\ }\ }+ \log  \det  \Bigg(I_{N_i}+{\rho }_{ii}H_{ii}Q_{ii}H^{\dagger }_{ii}-\left[ \begin{array}{cc}
\sqrt{{\rho }_{ii}{\rho }_{ij}}H_{ii}Q_{ii}H^{\dagger }_{ij} & \sqrt{{\rho }_{ii}}H_{ii}Q_{ij} \end{array}
\right]\nonumber\\
&&{\left[ \begin{array}{cc}
I_{N_j}+{\rho }_{ij}H_{ij}Q_{ii}H^{\dagger }_{ij} & \sqrt{{\rho }_{ij}}H_{ij}Q_{ij} \\
\sqrt{{\rho }_{ij}}Q_{ij}^{\dagger }H^{\dagger }_{ij} & I_{M_j} \end{array}
\right]}^{-1}\left[ \begin{array}{c}
\sqrt{{\rho }_{ii}{\rho }_{ij}}H_{ij}Q_{ii}H^{\dagger }_{ii} \\
\sqrt{{\rho }_{ii}}Q_{ij}^{\dagger }H^{\dagger }_{ii} \end{array}
\right]\Bigg).\label{tmqii}
\end{eqnarray}
Next, we will show that $Q_{ii}=I_{M_{i}}$ maximizes \eqref{tmqii}.

Let us define $S\triangleq\sqrt {{\rho }_{ij}}H^{\dagger}_{ij}$, $W\triangleq Q_{ii}-Q_{ij}Q^{\dagger}_{ij}$, $E\triangleq (I_{N_j}+S^{\dagger}WS)^{-1}$ and
\begin{eqnarray}
f\left(S,Q_{ii}\right)\triangleq Q_{ii}-\left[ \begin{array}{cc}
Q_{ii}S & Q_{ij} \end{array}
\right]{\left[ \begin{array}{cc}
I_{N_j}+S^{\dagger }Q_{ii}S & S^{\dagger }Q_{ij} \\
Q^{\dagger }_{ij}S & I_{M_j} \end{array}
\right]}^{-1}\left[ \begin{array}{c}
S^{\dagger }Q_{ii} \\
Q^{\dagger }_{ij} \end{array}
\right].
\end{eqnarray}

We can check that
\begin{eqnarray}\label{product2}
{\left[ \begin{array}{cc}
I_{N_j}+S^{\dagger }Q_{ii}S & S^{\dagger }Q_{ij} \\
Q^{\dagger }_{ij}S & I_{M_j} \end{array}
\right]}
{\left[ \begin{array}{cc}
E & -ES^{\dagger}Q_{ij} \\
-Q_{ij}^{\dagger}SE & I+Q_{ij}^{\dagger}SES^{\dagger}Q_{ij} \end{array}
\right]}=I_{M_j+N_j}.
\end{eqnarray}
Hence
\begin{eqnarray}
&&f\left(S,Q_{ii}\right)=Q_{ii}-\left[ \begin{array}{cc}
Q_{ii}S & Q_{ij} \end{array}
\right]{\left[ \begin{array}{cc}
I_{N_j}+S^{\dagger }Q_{ii}S & S^{\dagger }Q_{ij} \\
Q^{\dagger }_{ij}S & I_{M_j} \end{array}
\right]}^{-1}\left[ \begin{array}{c}
S^{\dagger }Q_{ii} \\
Q^{\dagger }_{ij} \end{array}
\right]\nonumber\\
&&=Q_{ii}-\left[ \begin{array}{cc}
Q_{ii}S & Q_{ij} \end{array}
\right]{\left[ \begin{array}{cc}
E & -ES^{\dagger}Q_{ij} \\
-Q_{ij}^{\dagger}SE & I+Q_{ij}^{\dagger}SES^{\dagger}Q_{ij} \end{array}
\right]}\left[ \begin{array}{c}
S^{\dagger }Q \\
Q^{\dagger }_{ij} \end{array}
\right]\nonumber\\
&&=Q_{ii}-Q_{ii}SES^{\dagger }Q_{ii}
+Q_{ii}SES^{\dagger}Q_{ij}Q^{\dagger }_{ij}
+Q_{ij}Q_{ij}^{\dagger}SES^{\dagger }Q_{ii}
-Q_{ij}Q^{\dagger }_{ij}-\nonumber\\
&&Q_{ij}Q_{ij}^{\dagger}SES^{\dagger}Q_{ij}Q^{\dagger }_{ij}\nonumber\\
&&=Q_{ii}-Q_{ij}Q^{\dagger}_{ij}-(Q_{ii}-Q_{ij}Q^{\dagger}_{ij})SES^{\dagger}(Q_{ii}-Q_{ij}Q^{\dagger}_{ij})\nonumber\\
&&=Q_{ii}-Q_{ij}Q^{\dagger}_{ij}-(Q_{ii}-Q_{ij}Q^{\dagger}_{ij})S(I+S^{\dagger}(Q_{ii}-Q_{ij}Q^{\dagger}_{ij})S)^{-1}S^{\dagger}(Q_{ii}-Q_{ij}Q^{\dagger}_{ij})\nonumber\\
&&=W-WS(I_{N_{j}}+S^{\dagger}WS)^{-1}S^{\dagger}W.\label{opt}
\end{eqnarray}

We know that $W=Q_{ii}-Q_{ij}Q^{\dagger}_{ij}\preceq I_{M_{i}}-Q_{ij}Q^{\dagger}_{ij}$. So, according to Lemma \ref{L} with $K_1$ as $Q_{ii}-Q_{ij}Q^{\dagger}_{ij}$ and $K_2$ as $I_{M_{i}}-Q_{ij}Q^{\dagger}_{ij}$, we have  $f\left(S,Q_{ii}\right)\preceq f(S,I_{M_i})$. Thus, we use this outer-bound by replacing $Q_{ii}$ by $I$ to get
\begin{eqnarray}
&&{\log  {\det  \left(K_{i3}\right)\ }\ }-{\log  {\det  \left(K_{i4}\right)\ }\ }\nonumber\\
&\le& \log  \det  \Bigg(I_{N_i}+{\rho }_{ii}H_{ii}H^{\dagger }_{ii}-\left[ \begin{array}{cc}
\sqrt{{\rho }_{ii}{\rho }_{ij}}H_{ii}H^{\dagger }_{ij} & \sqrt{{\rho }_{ii}}H_{ii}Q_{ij} \end{array}
\right]\nonumber\\
&&{\left[ \begin{array}{cc}
I_{N_j}+{\rho }_{ij}H_{ij}H^{\dagger }_{ij} & \sqrt{{\rho }_{ij}}H_{ij}Q_{ij} \\
\sqrt{{\rho }_{ij}}Q_{ij}^{\dagger }H^{\dagger }_{ij} & I_{M_j} \end{array}
\right]}^{-1}\left[ \begin{array}{c}
\sqrt{{\rho }_{ii}{\rho }_{ij}}H_{ij}H^{\dagger }_{ii} \\
\sqrt{{\rho }_{ii}}Q_{ij}^{\dagger }H^{\dagger }_{ii} \end{array}
\right]\Bigg).
\end{eqnarray}
Substituting this in \eqref{yixjsi}, we get
\begin{eqnarray}
h\left(Y_i\mathrel{\left|\vphantom{Y_i X_j,S_i}\right.\kern-\nulldelimiterspace}X_j,S_i\right)&\le&{\log  {\det  \left(K_{i3}\right)\ }\ }-{\log  {\det  \left(K_{i4}\right)\ }\ }+{\log  {\det  \pi e\left(I_{N_i}\right)\ }\ }\nonumber\\
&\le& \log  \det  \Bigg(I_{N_i}+{\rho }_{ii}H_{ii}H^{\dagger }_{ii}-\left[ \begin{array}{cc} \sqrt{{\rho }_{ii}{\rho }_{ij}}H_{ii}H^{\dagger }_{ij} & \sqrt{{\rho }_{ii}}H_{ii}Q_{ij} \end{array} \right]\nonumber\\
&&{\left[ \begin{array}{cc}
I_{N_j}+{\rho }_{ij}H_{ij}H^{\dagger }_{ij} & \sqrt{{\rho }_{ij}}H_{ij}Q_{ij} \\
\sqrt{{\rho }_{ij}}Q_{ij}^{\dagger }H^{\dagger }_{ij} & I_{M_j} \end{array}
\right]}^{-1}\left[ \begin{array}{c}
\sqrt{{\rho }_{ii}{\rho }_{ij}}H_{ij}H^{\dagger }_{ii} \\
\sqrt{{\rho }_{ii}}Q_{ij}^{\dagger }H^{\dagger }_{ii} \end{array}
\right]\Bigg)\nonumber\\
&&+{N_i}\log \left( \pi e\right).
\end{eqnarray}
\end{proof}

{ The rest of the section considers the 6 terms in Lemma \ref{lemma_conditional} and outer-bounds each of them to get the terms in the outer-bound of Theorem \ref{outer_capacity}.}

\noindent {\bf First term:} For the first term in Lemma \ref{lemma_conditional},
\begin{eqnarray}
R_1&\le& h(Y_1)-h(Z_1)\nonumber\\
&\stackrel{(a)}{\le}& {\log  {\det  \left(I_{N_1}+{\rho }_{11}H_{11}H^{\dagger }_{11}+{\rho }_{21}H_{21}H^{\dagger }_{21}+\sqrt{{\rho }_{11}{\rho }_{21}}H_{11}Q_{12}H^{\dagger }_{21}+\sqrt{{\rho }_{11}{\rho }_{21}}H_{21}Q_{12}^{\dagger }H^{\dagger }_{11}\right)\ }\ }\nonumber\\
&&+{N_1}\log \left( \pi e\right)-h(Z_1)\nonumber\\
&\stackrel{(b)}{=}& {\log  {\det  \left(I_{N_1}+{\rho }_{11}H_{11}H^{\dagger }_{11}+{\rho }_{21}H_{21}H^{\dagger }_{21}+\sqrt{{\rho }_{11}{\rho }_{21}}H_{11}Q_{12}H^{\dagger }_{21}+\sqrt{{\rho }_{11}{\rho }_{21}}H_{21}Q_{12}^{\dagger }H^{\dagger }_{11}\right),\ }\ }
\end{eqnarray}
where $(a)$ follows from Lemma \ref{lemma_hyi} and $(b)$ follows from the fact that $h(Z_1)=\log \det  \left( \pi e I_{N_1}\right)$.

\noindent {\bf Second term:} The second bound is similar to the first bound by exchanging $1$ and $2$ in the indices.

\noindent {\bf Third term:} For the  third bound in Lemma \ref{lemma_conditional}, it is sufficient to replace upper bounds of $h\left(Y_2\mathrel{\left|\vphantom{Y_2 X_2}\right.\kern-\nulldelimiterspace}X_2\right)$ and $h(Y_1|X_2,S_1)$ from Lemma \ref{yixi} and Lemma \ref{yixisi} as follows
\begin{eqnarray}
R_1&\le& h\left(Y_2\mathrel{\left|\vphantom{Y_2 X_2}\right.\kern-\nulldelimiterspace}X_2\right)-h\left(Z_2\right)+h(Y_1|X_2,S_1)-h(Z_1)\nonumber\\
&\stackrel{(a)}{\le}& {\log  {\det  \left(I_{N_2}+{\rho }_{12}H_{12}H^{\dagger }_{12}-{\rho }_{12}H_{12}Q_{21}^{\dagger }Q_{21}H^{\dagger }_{12}\right)\ }\ }+{N_2}\log \left( \pi e\right)\nonumber\\
&&+\log  \det  \Bigg(I_{N_1}+{\rho }_{11}H_{11}H^{\dagger }_{11}-\left[ \begin{array}{cc} \sqrt{{\rho }_{11}{\rho }_{12}}H_{11}H^{\dagger }_{12} & \sqrt{{\rho }_{11}}H_{11}Q_{12} \end{array} \right]\nonumber\\
&&{\left[ \begin{array}{cc}
I_{N_2}+{\rho }_{12}H_{12}H^{\dagger }_{12} & \sqrt{{\rho }_{12}}H_{12}Q_{12} \\
\sqrt{{\rho }_{12}}Q_{12}^{\dagger }H^{\dagger }_{12} & I_{M_2} \end{array}
\right]}^{-1}\left[ \begin{array}{c}
\sqrt{{\rho }_{11}{\rho }_{12}}H_{12}H^{\dagger }_{11} \\
\sqrt{{\rho }_{11}}Q_{12}^{\dagger }H^{\dagger }_{11} \end{array}
\right]\Bigg)\nonumber\\
&&+{N_1}\log \left( \pi e\right)-h(Z_1)-h(Z_2)\nonumber\\
&\stackrel{(b)}{=}& {\log  {\det  \left(I_{N_2}+{\rho }_{12}H_{12}H^{\dagger }_{12}-{\rho }_{12}H_{12}Q_{21}^{\dagger }Q_{21}H^{\dagger }_{12}\right)\ }\ }\nonumber\\
&&+\log  \det  \Bigg(I_{N_1}+{\rho }_{11}H_{11}H^{\dagger }_{11}-\left[ \begin{array}{cc} \sqrt{{\rho }_{11}{\rho }_{12}}H_{11}H^{\dagger }_{12} & \sqrt{{\rho }_{11}}H_{11}Q_{12} \end{array} \right]\nonumber\\
&&{\left[ \begin{array}{cc}
I_{N_2}+{\rho }_{12}H_{12}H^{\dagger }_{12} & \sqrt{{\rho }_{12}}H_{12}Q_{12} \\
\sqrt{{\rho }_{12}}Q_{12}^{\dagger }H^{\dagger }_{12} & I_{M_2} \end{array}
\right]}^{-1}\left[ \begin{array}{c}
\sqrt{{\rho }_{11}{\rho }_{12}}H_{12}H^{\dagger }_{11} \\
\sqrt{{\rho }_{11}}Q_{12}^{\dagger }H^{\dagger }_{11} \end{array}
\right]\Bigg),
\end{eqnarray}
where $(a)$ is obtained by using Lemma \ref{yixi} and Lemma \ref{yixisi} and $(b)$ follows from the fact that $h(Z_i)=\log \det \left( \pi e I_{N_i}\right)$, for $i=1,2$.

\noindent {\bf Fourth term:} The fourth term is similar to the third term by exchanging $1$ and $2$ in the indices.

\noindent {\bf Fifth term:}
 According to the fifth bound in Lemma \ref{lemma_conditional}, it is sufficient to replace upper bounds of $h(Y_1|X_2,S_1)$ and $h(Y_2)$ from from Lemma \ref{yixisi} and Lemma \ref{lemma_hyi}, respectively, and get the fifth bound of Theorem \ref{outer_capacity} as follows
\begin{eqnarray}
R_1+R_2 &\le& h\left(Y_1\mathrel{\left|\vphantom{Y_1 S_1,X_2}\right.\kern-\nulldelimiterspace}S_1,X_2\right)-h\left(Z_2\right)+h(Y_2)-h(Z_1)\nonumber\\
&\stackrel{(a)}{\le}& {\log  {\det  \left(I_{N_2}+{\rho }_{22}H_{22}H^{\dagger }_{22}+{\rho }_{12}H_{12}H^{\dagger }_{12}+\sqrt{{\rho }_{22}{\rho }_{12}}H_{22}Q_{12}H^{\dagger }_{12}+\sqrt{{\rho }_{22}{\rho }_{12}}H_{12}Q_{12}^{\dagger }H^{\dagger }_{22}\right)\ }\ }\nonumber\\
&&+{N_2}\log \left( \pi e\right)\nonumber\\
&&+\log  \det  \Bigg(I_{N_1}+{\rho }_{11}H_{11}H^{\dagger }_{11}-\left[ \begin{array}{cc} \sqrt{{\rho }_{11}{\rho }_{12}}H_{11}H^{\dagger }_{12} & \sqrt{{\rho }_{11}}H_{11}Q_{12} \end{array} \right]\nonumber\\
&&{\left[ \begin{array}{cc}
I_{N_2}+{\rho }_{12}H_{12}H^{\dagger }_{12} & \sqrt{{\rho }_{12}}H_{12}Q_{12} \\
\sqrt{{\rho }_{12}}Q_{12}^{\dagger }H^{\dagger }_{12} & I_{M_2} \end{array}
\right]}^{-1}\left[ \begin{array}{c}
\sqrt{{\rho }_{11}{\rho }_{12}}H_{12}H^{\dagger }_{11} \\
\sqrt{{\rho }_{11}}Q_{12}^{\dagger }H^{\dagger }_{11} \end{array}
\right]\Bigg)\nonumber\\
&&+{N_1}\log \left( \pi e\right)-h(Z_1)-h(Z_2)\nonumber\\
&\stackrel{(b)}{=}& {\log  {\det  \left(I_{N_2}+{\rho }_{22}H_{22}H^{\dagger }_{22}+{\rho }_{12}H_{12}H^{\dagger }_{12}+\sqrt{{\rho }_{22}{\rho }_{12}}H_{22}Q_{12}H^{\dagger }_{12}+\sqrt{{\rho }_{22}{\rho }_{12}}H_{12}Q_{12}^{\dagger }H^{\dagger }_{22}\right)\ }\ }\nonumber\\
&&+\log  \det  \Bigg(I_{N_1}+{\rho }_{11}H_{11}H^{\dagger }_{11}-\left[ \begin{array}{cc} \sqrt{{\rho }_{11}{\rho }_{12}}H_{11}H^{\dagger }_{12} & \sqrt{{\rho }_{11}}H_{11}Q_{12} \end{array} \right]\nonumber\\
&&{\left[ \begin{array}{cc}
I_{N_2}+{\rho }_{12}H_{12}H^{\dagger }_{12} & \sqrt{{\rho }_{12}}H_{12}Q_{12} \\
\sqrt{{\rho }_{12}}Q_{12}^{\dagger }H^{\dagger }_{12} & I_{M_2} \end{array}
\right]}^{-1}\left[ \begin{array}{c}
\sqrt{{\rho }_{11}{\rho }_{12}}H_{12}H^{\dagger }_{11} \\
\sqrt{{\rho }_{11}}Q_{12}^{\dagger }H^{\dagger }_{11} \end{array}
\right]\Bigg),
\end{eqnarray}
where $(a)$ is obtained by using Lemma \ref{yixisi} and Lemma \ref{lemma_hyi} and $(b)$ follows from the fact that $h(Z_i)=\log \det \left( 2
\pi e I_{N_i}\right)$, for $i=1,2$.

\noindent {\bf Sixth term:} The sixth term is similar to the fifth term by exchanging $1$ and $2$ in the indices.

\section{Proof of Achievability for Theorem \ref{inner_capacity}} \label{apdx_inner}

{ In this section, we prove the achievability for Theorem \ref{inner_capacity}. More precisely, we will show the following.}

\begin{lemma} \label{ach}
For a given set of $(\overline{H},\overline{\rho })$, the feedback capacity region of a two-user MIMO Gaussian IC can achieve all rate pairs $(R_1,R_2)\in \mathbb{A}(\overline{H},\overline{\rho })$ such that
\begin{eqnarray}
R_1&\le& {\log  {\det  (I_{N_1}+{\rho }_{11}H_{11}H^{\dagger }_{11}+{\rho }_{21}H_{21}H^{\dagger }_{21})}}-N_1,\\
R_2&\le& {\log  {\det  (I_{N_2}+{\rho }_{22}H_{22}H^{\dagger }_{22}+{\rho }_{12}H_{12}H^{\dagger }_{12})}}-N_2,\\
R_1&\le& \log  \det  \left(I_{N_2}+{\rho }_{12}H_{12}H^{\dagger }_{12}\right)+\log  \det (I_{N_1}+{\rho }_{11}H_{11}H^{\dagger }_{11}-\nonumber\\
&& \sqrt{{\rho }_{11}{\rho }_{12}}H_{11}H^{\dagger }_{12}
({I_{N_2}+{\rho }_{12}H_{12}H^{\dagger }_{12}})^{-1}
\sqrt{{\rho }_{11}{\rho }_{12}}H_{12}H^{\dagger }_{11})-N_1-N_2,\\
R_2&\le& \log  \det  \left(I_{N_1}+{\rho }_{21}H_{21}H^{\dagger }_{21}\right)
+\log \det (I_{N_2}+{\rho }_{22}H_{22}H^{\dagger }_{22}-\nonumber\\
&&\sqrt{{\rho }_{22}{\rho }_{21}}H_{22}H^{\dagger }_{21}
({I_{N_1}+{\rho }_{21}H_{21}H^{\dagger }_{21}})^{-1}
\sqrt{{\rho }_{22}{\rho }_{21}}H_{21}H^{\dagger }_{22})-N_1-N_2,\\
R_1+R_2&\le& \log \det  \left(I_{N_2}+{\rho }_{22}H_{22}H^{\dagger }_{22}+{\rho }_{12}H_{12}H^{\dagger }_{12}\right)
+\log  \det (I_{N_1}+{\rho }_{11}H_{11}H^{\dagger }_{11}-\nonumber\\
&& \sqrt{{\rho }_{11}{\rho }_{12}}H_{11}H^{\dagger }_{12}
({I_{N_2}+{\rho }_{12}H_{12}H^{\dagger }_{12}})^{-1}
\sqrt{{\rho }_{11}{\rho }_{12}}H_{12}H^{\dagger }_{11})-N_1-N_2,\\
R_1+R_2&\le& \log  \det  \left(I_{N_1}+{\rho }_{11}H_{11}H^{\dagger }_{11}+{\rho }_{21}H_{21}H^{\dagger }_{21}\right)
+\log \det (I_{N_2}+{\rho }_{22}H_{22}H^{\dagger }_{22}-\nonumber\\
&&\sqrt{{\rho }_{22}{\rho }_{21}}H_{22}H^{\dagger }_{21}
({I_{N_1}+{\rho }_{21}H_{21}H^{\dagger }_{21}})^{-1}
\sqrt{{\rho }_{22}{\rho }_{21}}H_{21}H^{\dagger }_{22})-N_1-N_2.
\end{eqnarray}
\end{lemma}

{ In order to prove this result, we will use the result in \cite{Tse} for a discrete memoryless channel. We will then give some Lemmas that would help in further inner-bounding these terms for a MIMO IC and finally go over each expression for the discrete memoryless channel to prove the result.}


\begin{lemma} \label{inner}
The feedback capacity region of the two-user discrete memoryless IC includes the set of $(R_1,R_2)$ such that \begin{eqnarray}
R_1&\le& I\left(U_2,X_1;Y_1\right),\\
R_2&\le& I\left(U_1,X_2;Y_2\right),\\
R_1&\le& I\left(U_1;Y_2|X_2\right)+I\left(X_1;Y_1\mathrel{\left|\vphantom{X_1;Y_1 U_1,U_2}\right.\kern-\nulldelimiterspace}U_1,U_2\right),\\
R_2&\le& I\left(U_2;Y_1|X_1\right)+I\left(X_2;Y_2\mathrel{\left|\vphantom{X_2;Y_2 U_1,U_2}\right.\kern-\nulldelimiterspace}U_1,U_2\right),\\
R_1+R_2&\le& I\left(X_1;Y_1\mathrel{\left|\vphantom{X_1;Y_1 U_1,U_2}\right.\kern-\nulldelimiterspace}U_1,U_2\right)+I\left(U_1,X_2;Y_2\right),\\
R_1+R_2&\le& I\left(X_2;Y_2\mathrel{\left|\vphantom{X_2;Y_2 U_1,U_2}\right.\kern-\nulldelimiterspace}U_1,U_2\right)+I\left(U_2,X_1;Y_1\right),
\end{eqnarray}
over all joint distributions $p(u_1)p(u_2)p(x_1|u_1)p(x_2|u_2)$.
\end{lemma}
\begin{proof}
{ This result is a special case of Lemma 1 of \cite{Tse}, obtained by substituting the auxiliary variable $U=0$.}
\end{proof}
To achieve this rate region, the authors of \cite{Tse} developed an infinite-staged achievable scheme that employs  block Markov encoding, backward decoding, and Han-Kobayashi message splitting.

The rest of the section inner bounds this region to get the inner bound in Theorem \ref{inner_capacity}. For this, we will introduce some useful lemmas.

\begin{lemma}\label{power}
The following holds for any  $M_i\times N_j$ matrix $S$
\begin{eqnarray}
 S({I_{N_j}+S^{\dagger }S})^{-1}S^{\dagger }\succeq 0.
\end{eqnarray}
\end{lemma}
\begin{proof}
It holds since it can be written as $AEA^{\dagger}$ for $A=S$ and $E={(I_{N_j}+S^{\dagger }S)}^{-1}$, which is p.s.d. because $E$ is p.s.d..
\end{proof}

\begin{lemma}\label{private}
The following holds for any  $M_i\times N_j$ matrix $S$
\begin{eqnarray}
\det (I_{N_j}+S^{\dagger }S-
S^{\dagger }S
(I_{N_j}+S^{\dagger }S)^{-1}
S^{\dagger }S)
\le 2^{N_j}.\label{pri}
\end{eqnarray}
\end{lemma}
\begin{proof}
Let us define $V\triangleq S^{\dagger }S$, we get
\begin{eqnarray}
&&{\det  (I_{N_j}+S^{\dagger }S-S^{\dagger }S{(I_{N_j}+S^{\dagger }S)}^{-1}S^{\dagger }S)\ }\nonumber\\
&=&{\det  (I_{N_j}+V-V{(I_{N_j}+V)}^{-1}V)\ }\nonumber\\
&=&{\det  (I_{N_j}+V-V{\left(I_{N_j}+V\right)}^{-1}(V+I_{N_j}-I_{N_j}))\ }\nonumber\\
&=&{\det  (I_{N_j}+V-V(I_{N_j}-{\left(I_{N_j}+V\right)}^{-1}))\ }\nonumber\\
&=&{\det  (I_{N_j}+V({\left(I_{N_j}+V\right)}^{-1}))\ }\nonumber\\
&=&{\det  (I_{N_j}+(-I_{N_j}+I_{N_j}+V)({\left(I_{N_j}+V\right)}^{-1}))\ }\nonumber\\
&=&{\det  \left(I_{N_j}+I_{N_j}-{\left(I_{N_j}+V\right)}^{-1}\right)\ }\nonumber\\
&\stackrel{(a)}{\le}& {\det  \left(2I_{N_j}\right)}\nonumber\\
&=&  2^{N_j},
\end{eqnarray}
where $(a)$  follows from the fact that $V=S^{\dagger }S$ is p.s.d., and its eigenvalues are non-negative. So, the eigenvalues of $I_{N_j}+V$ are greater than or equal to 1. As a result, eigenvalues of $(I_{N_j}+V)^{-1}$ are between 0 and 1, i.e. they satisfy $0\le\lambda_k\le1$. So
\begin{eqnarray}
\det (I_{N_j}+I_{N_j}-(I_{N_j}+V)^{-1})=(2-\lambda_1). ... .(2-\lambda_{N_j}) \le 2^{N_j},
\end{eqnarray}
which proves \eqref{pri}.
\end{proof}

As we said before, our achievability scheme has a power allocation according to \eqref{ipi} and \eqref{iui}. We note that this power allocation is feasible since $I_{M_i}-K_{X_{ip}}\succeq 0$ by Lemma \ref{power} substituting $\sqrt{{\rho }_{ij}}H^{\dagger }_{ij}$ into $S$.

We will now expand the achievability in Lemma \ref{inner} using  $U_i = X_{iu}$ for $i\in \{1,2\}$. Before expanding each term in Lemma \ref{inner}, we evaluate some entropies as follows.
\begin{eqnarray}\label{achyi}
&&h\left(Y_i\right)={\log  {\det (I_{N_i}+{\rho }_{ii}H_{ii}H^{\dagger }_{ii}+{\rho }_{ji}H_{ji}H^{\dagger }_{ji})\ }\ },
\end{eqnarray}
and
\begin{eqnarray}\label{achyixi}
h\left(Y_i|X_i\right)={\log  {\det \left(I_{N_i}+{\rho }_{ji}H_{ji}H^{\dagger }_{ji}\right)\ }\ }.
\end{eqnarray}
In addition, we have
\begin{eqnarray}\label{achyiu}
&&h\left(Y_i\mathrel{\left|\vphantom{Y_i U_i,U_j}\right.\kern-\nulldelimiterspace}U_i,U_j\right)\nonumber\\
&\ge& h(Y_i|U_i,U_j,X_j)\nonumber\\
&=&\log \det (I_{N_i}+{\rho }_{ii}H_{ii}K_{X_{ip}}H^{\dagger }_{ii})\nonumber\\
&=&\log  \det  (I_{N_i}+{\rho }_{ii}H_{ii}H^{\dagger }_{ii}-
\sqrt{{\rho }_{ii}{\rho }_{ij}}H_{ii}H^{\dagger }_{ij}
(I_{N_j}+{\rho }_{ij}H_{ij}H^{\dagger }_{ij})^{-1}
\sqrt{{\rho }_{ii}{\rho }_{ij}}H_{ij}H^{\dagger }_{ii}).
\end{eqnarray}
Moreover, we have
\begin{eqnarray}
h\left(Y_i\mathrel{\left|\vphantom{Y_i U_j,X_i}\right.\kern-\nulldelimiterspace}U_j,X_i\right)
&\le& {\log  {\det (I_{N_i}+{\rho }_{ji}H_{ji}K_{X_{jp}}H^{\dagger }_{ji})\ }\ }\nonumber\\
&\stackrel{(a)}{\le}& \log {\det  \left(2I_{N_i}\right)\ }\nonumber\\
&=& N_i,\label{U_B}
\end{eqnarray}
where $(a)$ follows from Lemma \ref{private} by substituting $\sqrt{{\rho }_{ji}}H_{ji}^{\dagger}$ in $S$. This shows that $h\left(Y_i\mathrel{\left|\vphantom{Y_i U_j,X_i}\right.\kern-\nulldelimiterspace}U_j,X_i\right)$ is upper-bounded by $N_i$.

In our achievability, $h\left(Y_i\mathrel{\left|\vphantom{Y_i U_j,X_i}\right.\kern-\nulldelimiterspace}U_j,X_i\right)$ appeared with a minus sign. So, without loss of generality we can replace it with its bound $N_i$ for the achievability.

The rest of the section considers the six terms in Lemma \ref{inner} and uses each of them to get the terms in the inner-bound of Lemma \ref{ach}.

\noindent {\bf First term:} For the first term in Lemma \ref{inner}, we have
\begin{eqnarray}
&&I\left(U_2,X_1;Y_1\right)\nonumber\\
&=&h\left(Y_1\right)-h(Y_1|U_2,X_1)\nonumber\\
&\stackrel{(a)}{=}&{\log  {\det (I_{N_1}+{\rho }_{11}H_{11}H^{\dagger }_{11}+{\rho }_{21}H_{21}H^{\dagger }_{21})}}-h(Y_1|U_2,X_1)\nonumber\\
&\stackrel{(b)}{\ge}&{\log  {\det (I_{N_1}+{\rho }_{11}H_{11}H^{\dagger }_{11}+{\rho }_{21}H_{21}H^{\dagger }_{21})}} -N_1,
\end{eqnarray}
where $(a)$ follows from \eqref{achyi} and $(b)$ follows from \eqref{U_B}.

\noindent {\bf Second term:} The second bound is similar to the first bound by exchanging $1$ and $2$ in the indices.

\noindent {\bf Third term:} For the third bound in Lemma \ref{inner}, we have
\begin{eqnarray}
&&I\left(U_1;Y_2|X_2\right)+I\left(X_1;Y_1\mathrel{\left|\vphantom{X_1;Y_1 U_1,U_2}\right.\kern-\nulldelimiterspace}U_1,U_2\right)\nonumber\\
&=&h\left(Y_2|X_2\right)-h\left(Y_2|U_1,X_2\right)+h\left(Y_1\mathrel{\left|\vphantom{Y_1 U_1,U_2}\right.\kern-\nulldelimiterspace}U_1,U_2\right)-h\left(Y_1\mathrel{\left|\vphantom{Y_1 U_1,U_2,X_1}\right.\kern-\nulldelimiterspace}U_1,U_2,X_1\right)\nonumber\\
&\ge& h\left(Y_2|X_2\right)-h\left(Y_2|{U_1,X}_2\right)+h(Y_1|U_1,U_2,X_2)-h\left(Y_1\mathrel{\left|\vphantom{Y_1 U_1,U_2,X_1}\right.\kern-\nulldelimiterspace}U_1,U_2,X_1\right)\nonumber\\
&\stackrel{(a)}{=}&{\log  {\det  \left(I_{N_2}+{\rho }_{12}H_{12}H^{\dagger }_{12}\right)}} +\log  \det  \Bigg(I_{N_1}+{\rho }_{11}H_{11}H^{\dagger }_{11}-
\sqrt{{\rho }_{11}{\rho }_{12}}H_{11}H^{\dagger }_{12}\nonumber\\
&&(I_{N_2}+{\rho }_{12}H_{12}H^{\dagger }_{12})^{-1}
\sqrt{{\rho }_{11}{\rho }_{12}}H_{12}H^{\dagger }_{11}\Bigg) -h\left(Y_2|U_1,X_2\right)-h\left(Y_1\mathrel{\left|\vphantom{Y_1 U_1,U_2,X_1}\right.\kern-\nulldelimiterspace}U_1,U_2,X_1\right)\nonumber\\
&\stackrel{(b)}{\ge}& {\log  {\det  \left(I_{N_2}+{\rho }_{12}H_{12}H^{\dagger }_{12}\right)}} +\log  \det  \Bigg(I_{N_1}+{\rho }_{11}H_{11}H^{\dagger }_{11}-
\sqrt{{\rho }_{11}{\rho }_{12}}H_{11}H^{\dagger }_{12}\nonumber\\
&&(I_{N_2}+{\rho }_{12}H_{12}H^{\dagger }_{12})^{-1}
\sqrt{{\rho }_{11}{\rho }_{12}}H_{12}H^{\dagger }_{11}\Bigg)- N_1 -N_2,
\end{eqnarray}
where $(a)$ is obtained from \eqref{achyixi} and \eqref{achyiu} and $(b)$ follows from (\ref{U_B}).

\noindent {\bf Fourth term:} The fourth term is similar to the third term by exchanging $1$ and $2$ in the indices.

\noindent {\bf Fifth term:} For the fifth bound in Lemma \ref{inner}, we have
\begin{eqnarray}
&&I\left(X_1;Y_1\mathrel{\left|\vphantom{X_1;Y_1 U_1,U_2}\right.\kern-\nulldelimiterspace}U_1,U_2\right)+I\left(U_1,X_2;Y_2\right)\nonumber\\
&=&h\left(Y_1\mathrel{\left|\vphantom{Y_1 U_1,U_2}\right.\kern-\nulldelimiterspace}U_1,U_2\right)-h\left(Y_1\mathrel{\left|\vphantom{Y_1 U_1,U_2,X_1}\right.\kern-\nulldelimiterspace}U_1,U_2,X_1\right)+h\left(Y_2\right)-h\left(Y_2|U_1,X_2\right)\\
&\ge& h(Y_1|U_1,U_2,X_2)-h\left(Y_1\mathrel{\left|\vphantom{Y_1 U_1,U_2,X_1}\right.\kern-\nulldelimiterspace}U_1,U_2,X_1\right)+h\left(Y_2\right)-h\left(Y_2|U_1,X_2\right)\\
&\stackrel{(a)}{=}&{\log  {\det  (I_{N_2}+{\rho }_{22}H_{22}H^{\dagger }_{22}+{\rho }_{12}H_{12}H^{\dagger }_{12})\ }\ }+\log  \det \Bigg(I_{N_1}+{\rho }_{11}H_{11}H^{\dagger }_{11}-
\sqrt{{\rho }_{11}{\rho }_{12}}H_{11}H^{\dagger }_{12}\nonumber\\
&&(I_{N_2}+{\rho }_{12}H_{12}H^{\dagger }_{12})^{-1}
\sqrt{{\rho }_{11}{\rho }_{12}}H_{12}H^{\dagger }_{11}\Bigg) -h\left(Y_2|{U_1,X}_2\right)-h\left(Y_1\mathrel{\left|\vphantom{Y_1 U_1,U_2,X_1}\right.\kern-\nulldelimiterspace}U_1,U_2,X_1\right)\\
&\stackrel{(b)}{\ge}&{\log  {\det  (I_{N_2}+{\rho }_{22}H_{22}H^{\dagger }_{22}+{\rho }_{12}H_{12}H^{\dagger }_{12})}}+\log  \det \Bigg(I_{N_1}+{\rho }_{11}H_{11}H^{\dagger }_{11}-
\sqrt{{\rho }_{11}{\rho }_{12}}H_{11}H^{\dagger }_{12}\nonumber\\
&&(I_{N_2}+{\rho }_{12}H_{12}H^{\dagger }_{12})^{-1}
\sqrt{{\rho }_{11}{\rho }_{12}}H_{12}H^{\dagger }_{11}\Bigg) - N_1 - N_2,
\end{eqnarray}

where $(a)$ is obtained from \eqref{achyi} and \eqref{achyiu}, and $(b)$ follows from \eqref{U_B}.

\noindent {\bf Sixth term:} The sixth term is similar to the fifth term by exchanging $1$ and $2$ in the indices.

\section{Proof of Outer Bound for Theorem 2}
\label{Appendix4}

{ In this section, we prove that covariance matrix $Q=0$ is approximately optimal for the capacity region of the MIMO IC with feedback.} As mentioned in Section \ref{main_results}, it is enough to prove that
\begin{eqnarray}
\mathcal{R}_o(Q)\subseteq \mathcal{R}_o(0)\oplus ([0,N_1]\times [0,N_2]),
\end{eqnarray}
for any covariance matrix $Q$.

Now, we give three important inequalities that would be used in the main proof.

Define $E\triangleq {(I_{N_2}+\sqrt{{\rho }_{ij}}H_{ij}(I-Q_{ij}Q_{ij}^{\dagger})\sqrt{{\rho }_{ij}}H_{ij}^\dagger)}^{-1}$). The first inequality is as follows
\begin{eqnarray}
&&I_{N_i}+{\rho }_{ii}H_{ii}H^{\dagger }_{ii}-\nonumber\\
&&\left[ \begin{array}{cc}
\sqrt{{\rho }_{ii}{\rho }_{ij}}H_{ii}H^{\dagger }_{ij} & \sqrt{{\rho }_{ii}}H_{ii}Q_{ij} \end{array}
\right]{\left[ \begin{array}{cc}
I_{N_j}+{\rho }_{ij}H_{ij}H^{\dagger }_{ij} & \sqrt{{\rho }_{ij}}H_{ij}Q_{ij} \\
\sqrt{{\rho }_{ij}}Q_{ij}^{\dagger }H^{\dagger }_{ij} & I_{M_j} \end{array}
\right]}^{-1}\left[ \begin{array}{c}
\sqrt{{\rho }_{ii}{\rho }_{ij}}H_{ij}H^{\dagger }_{ii} \\
\sqrt{{\rho }_{ii}}Q_{ij}^{\dagger }H^{\dagger }_{ii} \end{array}
\right]\nonumber\\
&=&I_{N_i}+{\rho }_{ii}H_{ii}\left(I_{M_i}-\right.\nonumber\\
&&\left.\left[ \begin{array}{cc}
\sqrt{{\rho }_{ij}}H^{\dagger }_{ij} & Q_{ij} \end{array}
\right]{\left[ \begin{array}{cc}
I_{N_j}+{\rho }_{ij}H_{ij}H^{\dagger }_{ij} & \sqrt{{\rho }_{ij}}H_{ij}Q_{ij} \\
\sqrt{{\rho }_{ij}}Q_{ij}^{\dagger }H^{\dagger }_{ij} & I_{M_j} \end{array}
\right]}^{-1}\left[ \begin{array}{c}
\sqrt{{\rho }_{ij}}H_{ij} \\
Q_{ij}^{\dagger } \end{array}
\right]\right)H^{\dagger }_{ii}\nonumber\\
&\stackrel{(a)}{=}&I_{N_i}+{\rho }_{ii}H_{ii}\left(I_{M_i}-\right.\nonumber\\
&&\left[ \begin{array}{cc}
\sqrt{{\rho }_{ij}}H^{\dagger }_{ij} & Q_{ij} \end{array}
\right]{\left[ \begin{array}{cc}
E & -E\sqrt{{\rho }_{ij}}H_{ij}Q_{ij} \\
-\sqrt{{\rho }_{ij}}Q_{ij}^{\dagger }H^{\dagger }_{ij}E & I_{M_j}+\sqrt{{\rho }_{ij}}Q_{ij}^{\dagger }H^{\dagger }_{ij}E\sqrt{{\rho }_{ij}}H_{ij}Q_{ij} \end{array}
\right]}\nonumber\\
&&\left.\left[ \begin{array}{c}
\sqrt{{\rho }_{ij}}H_{ij} \\
Q_{ij}^{\dagger } \end{array}
\right]\right)H^{\dagger }_{ii}\nonumber\\
&\stackrel{(b)}{=}&I_{N_i}+{\rho }_{ii}H_{ii}\left(I-Q_{ij}Q_{ij}^{\dagger}-(I-Q_{ij}Q_{ij}^{\dagger})\sqrt{{\rho }_{ij}}H_{ij}^\dagger E\sqrt{{\rho }_{ij}}H_{ij}(I-Q_{ij}Q_{ij}^{\dagger})\right)H^{\dagger }_{ii}\nonumber\\
&\stackrel{(c)}{=}&I_{N_i}+{\rho }_{ii}H_{ii}L\left(I-Q_{ij}Q_{ij}^{\dagger }, \sqrt{{\rho }_{ij}}H_{ij}^\dagger\right)H^{\dagger }_{ii}\nonumber\\
&\stackrel{(d)}{\le}&I_{N_i}+{\rho }_{ii}H_{ii}L\left(I, \sqrt{{\rho }_{ij}}H_{ij}^\dagger\right)H^{\dagger }_{ii},\label{m1}
\end{eqnarray}
where $L(K,S)$ is as in \eqref{eqdefnl}, $(a)$ follows since the inverse can be verified easily, $(b)$ follows from  finding the product of matrices, $(c)$ follows from the definition of $L(K,S)$ in \eqref{eqdefnl}, and (d) follows from Lemma \ref{L}. 

The second inequality is as follows

\begin{eqnarray}
&&{\log  {\det  \left(I_{N_j}+{\rho }_{ij}H_{ij}H^{\dagger }_{ij}-{\rho }_{ij}H_{ij}Q_{ij}Q_{ij}^{\dagger }H^{\dagger }_{ij}\right)\ }\ }\nonumber\\
&\le& {\log  {\det  \left(I_{N_j}+{\rho }_{ij}H_{ij}H^{\dagger }_{ij}\right)\ }\ }.\label{m2}
\end{eqnarray}

The third inequality is as follows
\begin{eqnarray}
&&{\log  {\det  (I_{N_i}+{\rho }_{ii}H_{ii}H^{\dagger }_{ii}+{\rho }_{ji}H_{ji}H^{\dagger }_{ji}+\sqrt{{\rho }_{ii}{\rho }_{ji}}H_{ii}Q_{ij}H^{\dagger }_{ji}+\sqrt{{\rho }_{ii}{\rho }_{ji}}H_{ji}Q_{ij}^{\dagger }H^{\dagger }_{ii})\ }\ }\nonumber\\
&\stackrel{(a)}{\le}&{\log  {\det  (I_{N_i}+{\rho }_{ii}H_{ii}H^{\dagger }_{ii}+{\rho }_{ji}H_{ji}H^{\dagger }_{ji}+{{\rho }_{ii}}H_{ii}Q_{ii}H^{\dagger }_{ii}+{{\rho }_{ji}}H_{ji}Q_{jj}H^{\dagger }_{ji})\ }\ }\nonumber\\
&\stackrel{(b)}{\le}&{\log  {\det  (I_{N_i}+2{\rho }_{ii}H_{ii}H^{\dagger }_{ii}+2{\rho }_{ji}H_{ji}H^{\dagger }_{ji})\ }\ }\nonumber\\
&\le& {\log  {\det  (I_{N_i}+{\rho }_{ii}H_{ii}H^{\dagger }_{ii}+{\rho }_{ji}H_{ji}H^{\dagger }_{ji})}}+N_i,\label{m3}
\end{eqnarray}
where $(a)$ follows from $\left(A-B\right)\left(A^{\dagger }-B^{\dagger }\right)=AA^{\dagger }+BB^{\dagger }-AB^{\dagger }-BA^{\dagger }\succeq 0$ by substituting $\sqrt{{\rho }_{ii}}H_{ii}X_i$ and $\sqrt{{\rho }_{ji}}H_{ji}X_j$ in $A$ and $B$, respectively, $(b)$ follows from the fact that $I \succeq Q_{ii}$.

Thus, we proved that among these three expansions, the first two expansions we started with are  maximized by $Q_{ij}=0$ while  the third one is is outer-bounded by the corresponding expression with $Q_{ij}=0$ plus $N_1$.

Now, we consider each of the six expressions in the definition of the region $\mathcal{R}_o(Q)$ and outer-bound each expression to find the gap with $\mathcal{R}_o(0)$ being constant thus proving that $\mathcal{R}_o(Q) \subseteq \mathcal{R}_o(0)\oplus ([0,N_1]\times [0,N_2])$ which proves the result.

Let the right-hand sides of the six expressions in the definition of $\mathcal{R}_0(Q)$ in \eqref{roqeq1}-\eqref{roqeql} be labeled as $I_1(Q)$, $I_2(Q)$, $I_3(Q)$, $I_4(Q)$, $I_5(Q)$, and $I_6(Q)$ respectively. Then, the constant gap outer-bound is shown in the following Lemma.

\begin{lemma}
We have
\begin{eqnarray}
I_1(Q)&\le& I_1(0)+N_1,\label{r1}\\
I_2(Q)&\le& I_2(0)+N_2,\label{r2}\\
I_3(Q)&\le& I_3(0),\label{r3}\\
I_4(Q)&\le& I_4(0),\label{r4}\\
I_5(Q)&\le& I_5(0)+N_2,\label{r5}\\
I_6(Q)&\le& I_6(0)+N_1.\label{r6}
\end{eqnarray}

\end{lemma}
\begin{proof}
We start with (\ref{r1}).
\begin{eqnarray}
I_1(Q)&=&{\log  {\det  (I_{N_1}+{\rho }_{11}H_{11}H^{\dagger }_{11}+{\rho }_{21}H_{21}H^{\dagger }_{21}+\sqrt{{\rho }_{11}{\rho }_{21}}H_{11}QH^{\dagger }_{21}+\sqrt{{\rho }_{11}{\rho }_{21}}H_{21}Q^{\dagger }H^{\dagger }_{11})\ }\ }\nonumber\\
&\stackrel{(a)}{\le}& {\log  {\det  (I_{N_1}+{\rho }_{11}H_{11}H^{\dagger }_{11}+{\rho }_{21}H_{21}H^{\dagger }_{21})\ }\ }+N_1\nonumber\\
&=& I_1(0)+N_1,
\end{eqnarray}
where $(a)$ follows from (\ref{m3}).

Proof of (\ref{r2}) is similar to (\ref{r1}) by exchanging $1$ and $2$ in the indices.

For the proof of (\ref{r3}) we have,
\begin{eqnarray}
I_3(Q)&=&
\log  \det  \left(I_{N_2}+{\rho }_{12}H_{12}H^{\dagger }_{12}-{\rho }_{12}H_{12}QQ^{\dagger }H^{\dagger }_{12}\right)+\log  \det \Bigg(I_{N_1}+{\rho }_{11}H_{11}H^{\dagger }_{11}-\nonumber\\
&&\left[ \begin{array}{cc}
\sqrt{{\rho }_{11}{\rho }_{12}}H_{11}H^{\dagger }_{12} & \sqrt{{\rho }_{11}}H_{11}Q \end{array}
\right]
{\left[ \begin{array}{cc}
I_{N_2}+{\rho }_{12}H_{12}H^{\dagger }_{12} & \sqrt{{\rho }_{12}}H_{12}Q \\
\sqrt{{\rho }_{12}}Q^{\dagger }H^{\dagger }_{12} & I_{M_2} \end{array}
\right]}^{-1}\nonumber\\
&&\left[ \begin{array}{c}
\sqrt{{\rho }_{11}{\rho }_{12}}H_{12}H^{\dagger }_{11} \\
\sqrt{{\rho }_{11}}Q^{\dagger }H^{\dagger }_{11} \end{array}
\right]\Bigg)\nonumber\\
&\stackrel{(a)}{\le}& \log  \det  \left(I_{N_2}+{\rho }_{12}H_{12}H^{\dagger }_{12}\right)+\log  \det \Bigg(I_{N_1}+{\rho }_{11}H_{11}H^{\dagger }_{11}-\nonumber\\
&&{\rho }_{11}{\rho }_{12}H_{11}H^{\dagger }_{12}
(I_{N_2}+{\rho }_{12}H_{12}H^{\dagger }_{12})^{-1}
H_{12}H^{\dagger }_{11}\Bigg)\nonumber\\
&=& I_3(0),
\end{eqnarray}
where $(a)$ follows since the first expression is outer-bounded as in  (\ref{m2}) and the outer-bound for the second expression can be shown on similar lines as  (\ref{m1}).

Proof of (\ref{r4}) is similar to (\ref{r3}) by exchanging $1$ and $2$ in the indices.

For the proof of (\ref{r5}) we have
\begin{eqnarray}
I_5(Q)&=& \log \det  \left(I_{N_2}+{\rho }_{22}H_{22}H^{\dagger }_{22}+{\rho }_{12}H_{12}H^{\dagger }_{12}+\sqrt{{\rho }_{22}{\rho }_{12}}H_{22}Q^{\dagger}H^{\dagger }_{12}+\sqrt{{\rho }_{22}{\rho }_{12}}H_{12}QH^{\dagger }_{22}\right)\nonumber\\
&&+\log\det  \Bigg(I_{N_1}+{\rho }_{11}H_{11}H^{\dagger }_{11}-\left[ \begin{array}{cc}
\sqrt{{\rho }_{11}{\rho }_{12}}H_{11}H^{\dagger }_{12} & \sqrt{{\rho }_{11}}H_{11}Q \end{array}
\right]\nonumber\\
&&{\left[ \begin{array}{cc}
I_{N_2}+{\rho }_{12}H_{12}H^{\dagger }_{12} & \sqrt{{\rho }_{12}}H_{12}Q \\
\sqrt{{\rho }_{12}}Q^{\dagger }H^{\dagger }_{12} & I_{M_2} \end{array}
\right]}^{-1}\left[ \begin{array}{c}
\sqrt{{\rho }_{11}{\rho }_{12}}H_{12}H^{\dagger }_{11} \\
\sqrt{{\rho }_{11}}Q^{\dagger }H^{\dagger }_{11} \end{array}
\right]\Bigg)\nonumber\\
&\stackrel{(a)}{\le}& {\log  {\det  (I_{N_2}+{\rho }_{22}H_{22}H^{\dagger }_{22}+{\rho }_{12}H_{12}H^{\dagger }_{12})}}+\log  \det \Bigg(I_{N_1}+{\rho }_{11}H_{11}H^{\dagger }_{11}-\nonumber\\
&&{\rho }_{11}{\rho }_{12}H_{11}H^{\dagger }_{12}
(I_{N_2}+{\rho }_{12}H_{12}H^{\dagger }_{12})^{-1}
H_{12}H^{\dagger }_{11}\Bigg)+N_2\nonumber\\
&=& I_5(0)+N_2,
\end{eqnarray}
where $(a)$ follows from (\ref{m3}) and using similar steps as in  (\ref{m1}).

Proof of (\ref{r6}) is similar to (\ref{r5}) by exchanging $1$ and $2$ in the indices.
\end{proof}

\section{Proof of Reciprocity in $\mathcal{R}_o(0)$}
\label{Appendix2}

In this section, we prove that replacing $\overline{H}$ and $\overline{\rho}$ by ${\overline{H}}^R$ and ${\overline{\rho}}^R$, respectively, and interchanging $M$ and $N$ for antennas at the nodes gives the same expressions in $\mathcal{R}_o(0)$.

We shall prove this in two steps. In the first step we shall prove
\begin{eqnarray}
\mathcal{R}_o(\overline{H},\overline{\rho})=\mathcal{R}_o({\overline{H}}^{'},{\overline{\rho}}^{R}),
\end{eqnarray}
where ${\overline{H}}^{'}=\{H^\dagger_{11},H^\dagger_{21},H^\dagger_{12},H^\dagger_{22}\}$ and in the second step we shall prove that
\begin{eqnarray}
\mathcal{R}_o({\overline{H}}^{'},\overline{\rho}^{R})=\mathcal{R}_o({\overline{H}}^{R},\overline{\rho}^{R}).
\end{eqnarray}
Clearly, the above two equalities prove the lemma.

Let the right-hand sides of the six expressions in the definition of $\mathcal{R}_0(0)$ in \eqref{ro0eq1}-\eqref{ro0eql} be labeled as $I_1$, $I_2$, $I_3$, $I_4$, $I_5$, and $I_6$ respectively.

{\bf First Step:} In this step, we prove that:
\begin{eqnarray}
I_1&=&I_3' \label{R1},\\
I_2&=&I_4' \label{R2},\\
I_3&=&I_1' \label{R3},\\
I_4&=&I_2' \label{R4},\\
I_5&=&I_6' \label{R5},\\
I_6&=&I_5' \label{R6},
\end{eqnarray}
where $I_k'$ is obtained from $I_k$ by interchanging $M$ and $N$, replacing $H_{ij}$ with $H_{ji}^\dagger$, and replacing $\rho_{ij}$ with $\rho_{ji}$.

Since $I_1$ and $I_3$ are both bounds for $R_1$,  $I_2$ and $I_4$ are both bounds for $R_2$, and $I_5$ and $I_6$ are both bounds for $R_1+R_2$, \eqref{R1}-\eqref{R6} will prove that $\mathcal{R}_o(\overline{H},\overline{\rho})=\mathcal{R}_o(\overline{H}^{'},\overline{\rho}^{R})$.

We start with proving (\ref{R1}). For simplicity we define $K\triangleq (I_{N_1}+{\rho }_{21}H_{21}H^{\dagger }_{21})^{-1}$, $K^{'}\triangleq (I_{M_1}+{\rho }_{21}H^{\dagger }_{21}H_{21})^{-1}$, and $L\triangleq {\rho }_{11}H_{11}H^{\dagger }_{11}$. We get
\begin{eqnarray}
I_1\label{rec1}
&=&{\log  {\det  (I_{N_1}+{\rho }_{11}H_{11}H^{\dagger }_{11}+{\rho }_{21}H_{21}H^{\dagger }_{21})}}\\
&=&{\log  {\det  (I_{N_1}+{\rho }_{21}H_{21}H^{\dagger }_{21})}}+
{\log  {\det  (I_{N_1}+K{\rho }_{11}H_{11}H^{\dagger }_{11})}}\nonumber\\
&=&{\log  {\det  (K^{-1})}}+
{\log  {\det  (I_{N_1}+KL)}}\nonumber\\
&\stackrel{(a)}{=}&{\log  {\det  (K^{-1})}}+
{\log  {\det  (I_{N_1}+LK)}}\nonumber\\
&=&{\log  {\det  (K^{-1})}}+
{\log  {\det  (I_{N_1}+LKI)}}\nonumber\\
&=&{\log  {\det  (K^{-1})}}+
{\log  {\det  (I_{N_1}+LK(I+{\rho }_{21}H_{21}H^{\dagger }_{21}-{\rho }_{21}H_{21}H^{\dagger }_{21}))}}\nonumber\\
&=&{\log  {\det  (K^{-1})}}+
{\log  {\det  (I_{N_1}+LK(I+{\rho }_{21}H_{21}H^{\dagger }_{21}-{\rho }_{21}H_{21}(I)H^{\dagger }_{21}))}}\nonumber\\
&=&{\log  {\det  (K^{-1})}}+
{\log  {\det  (I_{N_1}+LK(I+{\rho }_{21}H_{21}H^{\dagger }_{21}-{\rho }_{21}H_{21}({K^{'}}^{-1}K^{'})H^{\dagger }_{21}))}}\nonumber\\
&=&{\log  {\det  (K^{-1})}}+ \nonumber\\
&&{\log  {\det  (I_{N_1}+LK(I+{\rho }_{21}H_{21}H^{\dagger }_{21}-{\rho }_{21}H_{21}((I+{\rho }_{21}H^{\dagger }_{21}H_{21})K^{'})H^{\dagger }_{21}))}}\nonumber\\
&=&{\log  {\det  (K^{-1})}}+ \nonumber\\
&&{\log  {\det  (I_{N_1}+LK((I+{\rho }_{21}H_{21}H^{\dagger }_{21})-{\rho }_{21}((I+{\rho }_{21}H_{21}H^{\dagger }_{21})H_{21}K^{'})H^{\dagger }_{21}))}}\nonumber\\
&=&{\log  {\det  (K^{-1})}}+
{\log  {\det  (I_{N_1}+LK(K^{-1}-{\rho }_{21}K^{-1}H_{21}K^{'})H^{\dagger }_{21}))}}\nonumber\\
&=&{\log  {\det  (I+{\rho }_{21}H^{\dagger }_{21}H_{21})}}+
{\log  {\det  (I_{N_1}+L(I-{\rho }_{21}H_{21}K^{'})H^{\dagger }_{21}))}}\nonumber\\
&\stackrel{(b)}{=}&{\log  {\det  (I+{\rho }_{21}H^{\dagger }_{21}H_{21})}}+
{\log  {\det  (I+L-L{\rho }_{21}H_{21}
K^{'}H^{\dagger }_{21})}}\nonumber\\
&=&{\log  {\det  (I+{\rho }_{21}H^{\dagger }_{21}H_{21})}}+\nonumber\\
&&{\log  {\det  (I+{\rho }_{11}H_{11}H^{\dagger }_{11}-{\rho }_{11}{\rho }_{21}H_{11}H^{\dagger }_{11}H_{21}
(I+{\rho }_{21}H^{\dagger }_{21}H_{21})^{-1}H^{\dagger }_{21})}}\nonumber\\
&\stackrel{(c)}{=}&{\log  {\det  (I+{\rho }_{21}H^{\dagger }_{21}H_{21})}}+\nonumber\\
&&{\log  {\det  (I+{\rho }_{11}H^{\dagger }_{11}H_{11}-{\rho }_{11}{\rho }_{21}H^{\dagger }_{11}H_{21}
(I+{\rho }_{21}H^{\dagger }_{21}H_{21})^{-1}H^{\dagger }_{21}H_{11})}}\label{RR}\\
&=&I_3',
\end{eqnarray}
where $(a)$, $(b)$ and $(c)$ follow from Sylvester's determinant theorem \cite{det2}. (\ref{R2}) can be proved similarly due to symmetry. In addition, (\ref{R3}) and (\ref{R4}) can be obtained in the reverse direction similarly.

We move toward the proof of (\ref{R5}).
We should prove
\begin{eqnarray}
I_5&=&I_6',
\end{eqnarray}
where
\begin{eqnarray}
I_5
&=&{\log  {\det  (I_{N_2}+{\rho }_{22}H_{22}H^{\dagger }_{22}+{\rho }_{12}H_{12}H^{\dagger }_{12})}}+\nonumber\\
&&{\log  {\det  (I_{N_{1}}+{\rho }_{11}H_{11}H^{\dagger }_{11}-{\rho }_{11}{\rho }_{12}H_{11}H^{\dagger }_{12}
(I_{N_2}+{\rho }_{12}H_{12}H^{\dagger }_{12})^{-1}H_{12}H^{\dagger }_{11})}},
\end{eqnarray}
and
\begin{eqnarray}
I_6'
&=&{\log  {\det  (I_{M_1}+{\rho }_{11}H^{\dagger }_{11}H_{11}+{\rho }_{12}H^{\dagger }_{12}H_{12})}}+\nonumber\\
&&{\log  {\det  (I_{M{2}}+{\rho }_{22}H^{\dagger }_{22}H_{22}-{\rho }_{22}{\rho }_{12}H^{\dagger }_{22}H_{12}
(I_{M_{1}}+{\rho }_{12}H^{\dagger }_{12}H_{12})^{-1}H^{\dagger }_{12}H_{22})}}.
\end{eqnarray}
If we define
\begin{eqnarray}
a&\triangleq& {\log  {\det  (I_{N_2}+{\rho }_{22}H_{22}H^{\dagger }_{22}+{\rho }_{12}H_{12}H^{\dagger }_{12})}},\\
b&\triangleq& {\log  {\det  (I_{N_{1}}+{\rho }_{11}H_{11}H^{\dagger }_{11}-{\rho }_{11}{\rho }_{12}H_{11}H^{\dagger }_{12}
(I_{N_2}+{\rho }_{12}H_{12}H^{\dagger }_{12})^{-1}H_{12}H^{\dagger }_{11})}},\\
c&\triangleq& {\log  {\det  (I_{M_1}+{\rho }_{11}H^{\dagger }_{11}H_{11}+{\rho }_{12}H^{\dagger }_{12}H_{12})}},\\
d&\triangleq& {\log  {\det  (I_{M{2}}+{\rho }_{22}H^{\dagger }_{22}H_{22}-{\rho }_{22}{\rho }_{12}H^{\dagger }_{22}H_{12}
(I_{M_{1}}+{\rho }_{12}H^{\dagger }_{12}H_{12})^{-1}H^{\dagger }_{12}H_{22})}},
\end{eqnarray}
then, it is sufficient to prove $a+b=c+d$ or $a-d=c-b$.

Since \eqref{rec1} is equal to \eqref{RR}, we have
\begin{eqnarray}
&&{\log  {\det  (I_{N_1}+{\rho }_{11}H_{11}H^{\dagger }_{11}+{\rho }_{21}H_{21}H^{\dagger }_{21})}}-\nonumber\\
&&{\log  {\det  (I+{\rho }_{11}H^{\dagger }_{11}H_{11}-{\rho }_{11}{\rho }_{21}H^{\dagger }_{11}H_{21}
(I+{\rho }_{21}H^{\dagger }_{21}H_{21})^{-1}H^{\dagger }_{21}H_{11})}}=\nonumber\\
&&{\log  {\det  (I+{\rho }_{21}H^{\dagger }_{21}H_{21})}}.
\end{eqnarray}
Using similar method, we can see that
\begin{eqnarray}
a-d={\log  {\det  (I_{M_1}+{\rho }_{12}H^{\dagger }_{12}H_{12})}},
\end{eqnarray}
and
\begin{eqnarray}
c-b={\log  {\det  (I_{N_2}+{\rho }_{12}H_{12}H^{\dagger }_{12})}},
\end{eqnarray}
which according to Sylvester's determinant theorem \cite{det2}  are equal. This proves the $I_5=I_6'$.

(\ref{R6}) can be proved similar to the proof of \eqref{R5} due to symmetry.

{\bf Second Step:} It can be proved with a similar discussion as in Appendix E of \cite{Varanasi}. A brief sketch of the proof is given below for completeness.

Suppose $S$ is a p.s.d. matrix and $S^*$ represents its complex conjugate, i.e., the matrix obtained by replacing all its entries by the corresponding complex conjugates. Then, it is easy to see that
\begin{eqnarray}
{\log  {\det  (I+S)}}={\log  {\det  (I+S^*)}}.
\end{eqnarray}
However, note that all the terms in the different bounds of $\mathcal{R}_o(0)$ are of the form of ${\log  {\det  (I+S)}}$. This in turn proves that if we replace all the channel matrices of a two-user MIMO IC with feedback by their complex conjugates the set of upper bounds remain the same. From this fact, it easily follows that
\begin{eqnarray}
\mathcal{R}_o(\overline{H}^{'},\overline{\rho}^{R})=\mathcal{R}_o(\overline{H}^{R},\overline{\rho}^{R}).
\end{eqnarray}

\section{Proof of Theorem \ref{thm_gdof}}
\label{Appendix3}

In this section, we will find the limit of $\mathcal{R}_o(0)/\log \mathsf{SNR}$ as $\mathsf{SNR}\to \infty$ to get the result as in the statement of the Theorem \ref{thm_gdof} when ${{\rho}_{ij}}\sim{\mathsf{SNR}^{\alpha_{ij}}}$ (${{\rho}_{ij}}\sim{\mathsf{SNR}^{\alpha_{ij}}}$ represents that $ \lim_{\mathsf{SNR}\to\infty} \frac{\log {\rho}_{ij}}{\log \mathsf{SNR}} = \alpha_{ij}$). This follows from Theorem \ref{outer_inner_capacity_reciprocal} since the capacity region is inner and outer- bounded by  $\mathcal{R}_o(0)$ with constant gaps which would vanish for the degrees of freedom.

Before going over each of the terms in $\mathcal{R}_o(0)$ and finding its high SNR limit, we first give some Lemmas that will be used for the proof of the Theorem.


\begin{lemma}[\cite{Varanasi}] Let $H_{ij}\in \mathbb{C}^{N_j\times M_i}$ be a full rank channel matrix. Then, the following holds
\begin{eqnarray}\label{dof3}
{\log  {\det  \left(I_{N_j}+{\rho }_{ij}H_{ij}H^{\dagger }_{ij}\right)\ }\ } = {\alpha }_{ij}{\min  \left(M_i,N_j\right)\ }{\log \mathsf{SNR}\ }+o({\log \mathsf{SNR}}),
\end{eqnarray}
where ${{\rho}_{ij}}\sim{\mathsf{SNR}^{\alpha_{ij}}}$.
\end{lemma}

\begin{lemma}[\cite{Varanasi}]
Let $H_{ii}\in \mathbb{C}^{N_i\times M_i}$ and $H_{ji}\in \mathbb{C}^{N_i\times M_j}$ be two full rank channel matrices such that $[H_{ii} H_{ji}]$ is also full rank. Then, the following holds
\begin{eqnarray}
{\log  {\det  (I_{N_i}+{\rho }_{ii}H_{ii}H^{\dagger }_{ii}+{\rho }_{ji}H_{ji}H^{\dagger }_{ji})\ }\ } =f(N_i,\left({\alpha }_{ii},M_i\right),\left({\alpha }_{ji},M_j\right)){\log \mathsf{SNR}\ }+o({\log \mathsf{SNR}})
\end{eqnarray}\label{dof2}
where $f$ is defined in \eqref{f} and ${{\rho}_{ij}}\sim{\mathsf{SNR}^{\alpha_{ij}}}$.
\end{lemma}

\begin{lemma} Let $\Sigma\in \mathbb{C}^{N\times M}$ be a diagonal matrix with elements $\sigma_1,...,\sigma_m$ where $m=min(M,N)$ and $\Lambda\in \mathbb{C}^{m\times m}$ be a diagonal matrix with elements $|{\sigma_1|}^2,...,|{\sigma_m}|^2$, then
\begin{eqnarray}\label{diagonal}
&&{\Sigma }^{\dagger }\left[ \begin{array}{cc}
{(I_{m}+{\Lambda })}^{-1} & 0 \\
0 & I_{{(N-M)}^+} \end{array}
\right]{\Sigma }\nonumber\\
&=&\left[ \begin{array}{cc}
I_{m}-{(I_{m}+{\Lambda })}^{-1} & 0 \\
0 & 0_{(M-N)^+} \end{array}
\right].
\end{eqnarray}
\end{lemma}
\begin{proof}
We will split the proof in two cases, depending on whether $M\ge N$ or $M<N$.

\noindent {\bf Case 1 - $M\ge N$: } In this case, we have
\begin{eqnarray}
&&{\Sigma }^{\dagger }\left[ \begin{array}{cc}
{(I_{m}+{\Lambda })}^{-1} & 0 \\
0 & I_{{(N-M)}^+} \end{array}
\right]{\Sigma }\nonumber\\
&=&\left[ \begin{array}{ccc}
\sigma_1^* & 0 & 0 \\
0 & \ddots & 0 \\
0 & 0 & \sigma_m^* \\
0 & 0 & 0 \end{array}
\right]
\left[ \begin{array}{ccc}
\frac{1}{1+|\sigma_1|^2} & 0 & 0 \\
0 & \ddots & 0 \\
0 & 0 & \frac{1}{1+|\sigma_m|^2} \end{array}
\right]
\left[ \begin{array}{cccc}
\sigma_1 & 0 & 0 & 0 \\
0 & \ddots & 0 &0 \\
0 & 0 & \sigma_m & 0 \end{array}
\right]\nonumber\\
&=&
\left[ \begin{array}{cccc}
\frac{|\sigma_1|^2}{1+|\sigma_1|^2} & 0 & 0 & 0 \\
0 & \ddots & 0 &0 \\
0 & 0 & \frac{|\sigma_m|^2}{1+|\sigma_m|^2} & 0 \\
0 & 0 & 0 & 0_{(M-N)^+} \end{array}
\right]\nonumber\\
&=&
\left[ \begin{array}{cccc}
1-\frac{1}{1+|\sigma_1|^2} & 0 & 0 & 0 \\
0 & \ddots & 0 &0 \\
0 & 0 & 1-\frac{1}{1+|\sigma_m|^2} & 0 \\
0 & 0 & 0 & 0_{(M-N)^+} \end{array}
\right]\nonumber\\
&=&\left[ \begin{array}{cc}
I_{m}-{(I_{m}+{\Lambda })}^{-1} & 0 \\
0 & 0_{(M-N)^+} \end{array}
\right].
\end{eqnarray}

\noindent {\bf Case 2 - $M <N$: } In this case, we have
\begin{eqnarray}
&&{\Sigma }^{\dagger }\left[ \begin{array}{cc}
{(I_{m}+{\Lambda })}^{-1} & 0 \\
0 & I_{{(N-M)}^+} \end{array}
\right]{\Sigma }\nonumber\\
&=&\left[ \begin{array}{cccc}
\sigma_1^* & 0 & 0 & 0 \\
0 & \ddots & 0 &0 \\
0 & 0 & \sigma_m^* & 0 \end{array}
\right]
\left[ \begin{array}{cccc}
\frac{1}{1+|\sigma_1|^2} & 0 & 0 & 0 \\
0 & \ddots & 0 & 0 \\
0 & 0 & \frac{1}{1+|\sigma_m|^2} & 0 \\
0 & 0 & 0 & I_{(N-M)^+} \end{array}
\right]
\left[ \begin{array}{ccc}
\sigma_1 & 0 & 0 \\
0 & \ddots & 0 \\
0 & 0 & \sigma_m \\
0 & 0 & 0 \end{array}
\right]\nonumber\\
&=&
\left[ \begin{array}{cccc}
\frac{|\sigma_1|^2}{1+|\sigma_1|^2} & 0 & 0 & 0 \\
0 & \ddots & 0 &0 \\
0 & 0 & \frac{|\sigma_m|^2}{1+|\sigma_m|^2} & 0 \\
0 & 0 & 0 & 0_{(M-N)^+} \end{array}
\right]\nonumber\\
&=&
\left[ \begin{array}{cccc}
1-\frac{1}{1+|\sigma_1|^2} & 0 & 0 & 0 \\
0 & \ddots & 0 &0 \\
0 & 0 & 1-\frac{1}{1+|\sigma_m|^2} & 0 \\
0 & 0 & 0 & 0_{(M-N)^+} \end{array}
\right]\nonumber\\
&=&\left[ \begin{array}{cc}
I_{m}-{(I_{m}+{\Lambda })}^{-1} & 0 \\
0 & 0_{(M-N)^+} \end{array}
\right].
\end{eqnarray}
\end{proof}

\begin{lemma}\label{lemma-dof}
Let $H_{ii}\in \mathbb{C}^{N_i\times M_i}$ and $H_{ij}\in \mathbb{C}^{N_i\times M_j}$ be two channel matrices with each entry independently chosen from $\mathsf{CN}(0,1)$. Then, the following holds with probability $1$ (over the randomness of channel matrices).
\begin{eqnarray}
&&\log \det (I_{N_i}+{\rho }_{ii}H_{ii}H^{\dagger }_{ii}-\sqrt{{\rho }_{ii}{\rho }_{ij}}H_{ii}H^{\dagger }_{ij}{(I_{N_j}+{\rho }_{ij}H_{ij}H^{\dagger }_{ij})}^{-1}
\sqrt{{\rho }_{ii}{\rho }_{ij}}H_{ij}H^{\dagger }_{ii})\nonumber\\
&&=\left[{\alpha }_{ii}{\min  \left({\left(M_i-N_j\right)}^+,N_i\right)\ }+{\left({\alpha }_{ii}-{\alpha }_{ij}\right)}^+\left({\min  \left(M_i,N_i\right)\ }-{\min  \left({\left(M_i-N_j\right)}^+,N_i\right)\ }\right)\right]{\log \mathsf{SNR}\ }\nonumber\\
&&+o({\log  \mathsf{SNR}}).
\end{eqnarray}
where ${{\rho}_{ij}}\sim{\mathsf{SNR}^{\alpha_{ij}}}$.
\end{lemma}
\begin{proof}
Let the singular value decomposition (SVD) of the channel matrix $H_{ij}$ be given by $H_{ij}=V_{ij}$ ${\Sigma }_{ij}U^{\dagger }_{ij}$, where $V_{ij}\in U^{N_j\times N_j}$ and $U_{ij}\in U^{M_i\times M_i}$ are unitary matrices and ${\Sigma }_{ij}\in U^{N_j\times M_i}$ is a rectangular matrix containing the singular values along its diagonal. Using the SVD of the matrix $H_{ij}$ we get
\begin{eqnarray}
&&I_{N_i}+{\rho }_{ii}H_{ii}H^{\dagger }_{ii}-
\sqrt{{\rho }_{ii}{\rho }_{ij}}H_{ii}H^{\dagger }_{ij}(I_{N_j}+{\rho }_{ij}H_{ij}H^{\dagger }_{ij})^{-1}
\sqrt{{\rho }_{ii}{\rho }_{ij}}H_{ij}H^{\dagger }_{ii}\nonumber\\
&=&I_{N_i}+{\rho }_{ii}H_{ii}\left(I_{M_i}-{\rho }_{ij}H^{\dagger }_{ij}{\left(I_{N_j}+{\rho }_{ij}H_{ij}H^{\dagger }_{ij}\right)}^{-1}H_{ij}\right)H^{\dagger }_{ii}\\
&\stackrel{(a)}{=}&I_{N_i}+{\rho }_{ii}H_{ii}(I_{M_i}-{\rho }_{ij}H^{\dagger }_{ij}V_{ij}\left[ \begin{array}{cc}
{(I_{m_{ij}}+{\mathsf{SNR}}^{{\alpha }_{ij}}{\Lambda }_{ij})}^{-1} & 0 \\
0 & I_{{(N_j-M_i)}^+} \end{array}
\right]V^{\dagger }_{ij}H_{ij})H^{\dagger }_{ii}\nonumber\\
&=&I_{N_i}+{\rho }_{ii}H_{ii}(I_{M_i}-{\rho }_{ij}U_{ij}{\Sigma }^{\dagger }_{ij}\left[ \begin{array}{cc}
{(I_{m_{ij}}+{\mathsf{SNR}}^{{\alpha }_{ij}}{\Lambda }_{ij})}^{-1} & 0 \\
0 & I_{{(N_j-M_i)}^+} \end{array}
\right]{\Sigma }_{ij}U^{\dagger }_{ij})H^{\dagger }_{ii}\nonumber\\
&=&I_{N_i}+{\rho }_{ii}H_{ii}U_{ij}(I_{M_i}-{\mathsf{SNR}}^{{\alpha }_{ij}}{\Sigma }^{\dagger }_{ij}\left[ \begin{array}{cc}
{(I_{m_{ij}}+{\mathsf{SNR}}^{{\alpha }_{ij}}{\Lambda }_{ij})}^{-1} & 0 \\
0 & I_{{(N_j-M_i)}^+} \end{array}
\right]{\Sigma }_{ij})U^{\dagger }_{ij}H^{\dagger }_{ii}\nonumber\\
&\stackrel{(b)}{=}&I_{N_i}+{\rho }_{ii}H_{ii}U_{ij}(I_{M_i}-\left[ \begin{array}{cc}
I_{m_{ij}}-{(I_{m_{ij}}+{\mathsf{SNR}}^{{\alpha }_{ij}}{\Lambda }_{ij})}^{-1} & 0 \\
0 & 0_{(M_i-N_j)^+} \end{array}
\right])U^{\dagger }_{ij}H^{\dagger }_{ii}\nonumber\\
&=&I_{N_i}+{\mathsf{SNR}}^{{\alpha }_{ii}}H_{ii}U_{ij}\left[ \begin{array}{cc}
{(I_{m_{ij}}+{\mathsf{SNR}}^{{\alpha }_{ij}}{\Lambda }_{ij})}^{-1} & 0 \\
0 & I_{{(M_i-N_j)}^+} \end{array}
\right]U^{\dagger }_{ij}H^{\dagger }_{ii},
\end{eqnarray}
where $(a)$ results from SVD of the matrix $H_{ij}$ and $(b)$ follows from Lemma \ref{diagonal}.

Let us  decompose $U_{ij}\in U^{M_i\times M_i}$ into two parts,  $U_{ij1}$ and $U_{ij2}$ such that $U_{ij}=[U_{ij1}\ U_{ij2}]$, where $U_{ij1}\in U^{M_i\times {\min  \{M_i,N_j\}\ }}$ and $U_{ij2}\in U^{M_i\times (M_i-N_j)^+}$. Then, we get

\begin{eqnarray}
&&\log \det (I_{N_i}+{\rho }_{ii}H_{ii}H^{\dagger }_{ii}-
\sqrt{{\rho }_{ii}{\rho }_{ij}}H_{ii}H^{\dagger }_{ij}(I_{N_j}+{\rho }_{ij}H_{ij}H^{\dagger }_{ij})^{-1}
\sqrt{{\rho }_{ii}{\rho }_{ij}}H_{ij}H^{\dagger }_{ii})\nonumber\\
&=&\log \det (I_{N_i}+{\mathsf{SNR}}^{{\alpha }_{ii}}H_{ii}(U_{ij}\left[ \begin{array}{cc}
{(I_{m_{ij}}+{\mathsf{SNR}}^{{\alpha }_{ij}}{\Lambda }_{ij})}^{-1} & 0 \\
0 & I_{{(M_i-N_j)}^+} \end{array}
\right]U^{\dagger }_{ij})H^{\dagger }_{ii})\nonumber\\
&=&\log \det (I_{N_i}+H_{ii}({\mathsf{SNR}}^{{\alpha }_{ii}}U_{ij1}{(I_{m_{ij}}+{\mathsf{SNR}}^{{\alpha }_{ij}}{\Lambda }_{ij})}^{-1}U^{\dagger }_{ij1}+{\mathsf{SNR}}^{{\alpha }_{ii}}(U_{ij2}U^{\dagger }_{ij2}))H^{\dagger }_{ii})\nonumber\\
&=&\log \det (I_{N_i}+{\mathsf{SNR}}^{{\alpha }_{ii}}H_{ii}U_{ij2}U^{\dagger }_{ij2}H^{\dagger }_{ii}+{\mathsf{SNR} }^{{{\alpha }_{ii}}-{{{\alpha }_{ij}}}}H_{ii}U_{ij1}{({\mathsf{SNR} }^{{-\alpha }_{ij}}I_{m_{ij}}+{\Lambda }_{ij})}^{-1}U^{\dagger }_{ij1}H^{\dagger }_{ii},
\end{eqnarray}
where $m_{ij}=\min(M_i,N_j)$, $\Lambda_{ij}$ is a diagonal matrix containing the non-zero eigenvalues of $H_{ij}H_{ij}^{\dagger}$.

We note that ${\Lambda }_{ij}$ is invertible and when $\mathsf{SNR}$ is large, we can bound ${\mathsf{SNR} }^{{-\alpha }_{ij}}I_{m_{ij}}+{\Lambda }_{ij}$ from above and below as,  ${\Lambda }_{ij}\preceq{\mathsf{SNR} }^{{-\alpha }_{ij}}I_{m_{ij}}+{\Lambda }_{ij}\preceq I+{\Lambda }_{ij}$. We will only pursue the direction where ${\mathsf{SNR} }^{{-\alpha }_{ij}}I_{m_{ij}}+{\Lambda }_{ij}\succeq {\Lambda }_{ij}$ and can see that both the directions produce the same result and thus replacing the inner and outer bound by equality. In what follows, even though ${\mathsf{SNR} }^{{-\alpha }_{ij}}I_{m_{ij}}+{\Lambda }_{ij}\succeq {\Lambda }_{ij}$, we will substitute ${\mathsf{SNR} }^{{-\alpha }_{ij}}I_{m_{ij}}+{\Lambda }_{ij} = {\Lambda }_{ij}$ since by the inner and outer-bounding approach, it can be seen that the limit will be exactly the same thus not causing any difference in the result. Thus, we have

\begin{eqnarray}
&&\log \det (I_{N_i}+{\rho }_{ii}H_{ii}H^{\dagger }_{ii}-
\sqrt{{\rho }_{ii}{\rho }_{ij}}H_{ii}H^{\dagger }_{ij}(I_{N_j}+{\rho }_{ij}H_{ij}H^{\dagger }_{ij})^{-1}
\sqrt{{\rho }_{ii}{\rho }_{ij}}H_{ij}H^{\dagger }_{ii})\nonumber\\
&=&\log \det (I_{N_i}+{\mathsf{SNR}}^{{\alpha }_{ii}}H_{ii}U_{ij2}U^{\dagger }_{ij2}H^{\dagger }_{ii}+{\mathsf{SNR} }^{({{\alpha }_{ii}}-{{{\alpha }_{ij}}})}H_{ii}U_{ij1}{({\Lambda }_{ij})}^{-1}U^{\dagger }_{ij1}H^{\dagger }_{ii} + o({\log \mathsf{SNR}})\nonumber\\
&\stackrel{(a)}{=}&\log \det (I_{N_i}+{\mathsf{SNR}}^{{\alpha }_{ii}}H_{ii}U_{ij2}U^{\dagger }_{ij2}H^{\dagger }_{ii}+{\mathsf{SNR} }^{({{\alpha }_{ii}}-{{{\alpha }_{ij}}})^+}H_{ii}U_{ij1}{({\Lambda }_{ij})}^{-1}U^{\dagger }_{ij1}H^{\dagger }_{ii}\nonumber\\
&\stackrel{(b)}{=}&f(N_i,({\alpha }_{ii},(M_i-N_j)^+),(({\alpha }_{ii}-{\alpha }_{ij})^+,min(M_i,N_j))){\log \mathsf{SNR}\ }+o({\log \mathsf{SNR}\ })\nonumber\\
&=&[{\alpha }_{ii}{\min  \left({\left(M_i-N_j\right)}^+,N_i\right)\ }+{\left({\alpha }_{ii}-{\alpha }_{ij}\right)}^+\min((N_i-(M_i-N_j)^+)^+,N_j,M_i)
{\log \mathsf{SNR}\ }\nonumber\\
&&+o({\log \mathsf{SNR}\ })\nonumber\\
&\stackrel{(c)}{=}&\left[{\alpha }_{ii}{\min  \left({\left(M_i-N_j\right)}^+,N_i\right)\ }+{\left({\alpha }_{ii}-{\alpha }_{ij}\right)}^+\left({\min  \left(M_i,N_i\right)\ }-{\min  \left({\left(M_i-N_j\right)}^+,N_i\right)\ }\right)\right]{\log \mathsf{SNR}\ }\nonumber\\
&&+o({\log \mathsf{SNR}\ }),
\end{eqnarray}
where $(a)$ follows from the fact that if ${({{\alpha }_{ii}}-{{{\alpha }_{ij}}})}$ is less than zero we have
\begin{eqnarray}
{\mathsf{SNR} }^{({{\alpha }_{ii}}-{{{\alpha }_{ij}}})^+}H_{ii}U_{ij1}{({\Lambda }_{ij})}^{-1}U^{\dagger }_{ij1}H^{\dagger }_{ii}=o({\log \mathsf{SNR}}),
\end{eqnarray}
$(b)$ follows from Lemma \ref{dof2} and that $H_{ii}U_{ij1}$, $H_{ii}U_{ij1}\Lambda_{ij}^{-1/2}$ and $H_{ii}[U_{ij2}\ U_{ij1}\Lambda_{ij}^{-1/2}]$ are all full rank with probability 1; $(c)$ follows from some simple manipulations.
\end{proof}

The rest of the section considers the 6 terms in $\mathcal{R}_o(0)$ in \eqref{ro0eq1}-\eqref{ro0eql}, and finds the GDoF region for the MIMO IC with feedback.

\noindent {\bf First term:} According to the first bound in $\mathcal{R}_o(0)$, we have
\begin{eqnarray}
&&{\log  {\det  (I_{N_1}+{\rho }_{11}H_{11}H^{\dagger }_{11}+{\rho }_{21}H_{21}H^{\dagger }_{21})\ }\ }\nonumber\\
&\stackrel{(a)}{=}&f((N_1,\left({\alpha }_{11},M_1\right),\left({\alpha }_{21},M_2\right))){\log  \mathsf{SNR}\ }+o({\log  \log  \mathsf{SNR}\ }),
\end{eqnarray}
where $(a)$ is obtained from (\ref{dof2}). Now, dividing both sides by $\log  \mathsf{SNR}$, we get the first GDoF expression.

\noindent {\bf Second term:} The second bound is similar to the first bound by exchanging $1$ and $2$ in the indices.

\noindent {\bf Third term:} According to the third bound in $\mathcal{R}_o(0)$, we have
\begin{eqnarray}
&&\log  \det  \left(I_{N_2}+{\rho }_{12}H_{12}H^{\dagger }_{12}\right)+\log  \det (I_{N_1}+{\rho }_{11}H_{11}H^{\dagger }_{11}-\nonumber\\
&&\sqrt{{\rho }_{11}{\rho }_{12}}H_{11}H^{\dagger }_{12}(I_{N_2}+{\rho }_{12}H_{12}H^{\dagger }_{12})^{-1}
\sqrt{{\rho }_{11}{\rho }_{12}}H_{12}H^{\dagger }_{11})\nonumber\\
&\stackrel{(a)}{=}&{\alpha }_{12}{\min  \left(M_1,N_2\right)\ }+{\alpha }_{11}{\min  \left({\left(M_1-N_2\right)}^+,N_1\right)\ }+\nonumber\\
&&{\left({\alpha }_{11}-{\alpha }_{12}\right)}^+{\min  \left(M_1,N_1\right)\ }-{\min  \left({\left(M_1-N_2\right)}^+,N_1\right)\ }
+o({\log \mathsf{SNR}\ }),
\end{eqnarray}
where $(a)$ is obtained from Lemma \ref{dof3} and Lemma \ref{lemma-dof}. Now, dividing both sides by $\log \mathsf{SNR}$, the third GDoF bound results.

\noindent {\bf Fourth term:} The fourth term is similar to the third term by exchanging $1$ and $2$ in the indices.

\noindent {\bf Fifth term:} According to the fifth bound in $R_o(0)$, we have
\begin{eqnarray}
&&\log \det  \left(I_{N_2}+{\rho }_{22}H_{22}H^{\dagger }_{22}+{\rho }_{12}H_{12}H^{\dagger }_{12}\right)\nonumber\\
&&+\log\det (I_{N_1}+{\rho }_{11}H_{11}H^{\dagger }_{11}-
\sqrt{{\rho }_{11}{\rho }_{12}}H_{11}H^{\dagger }_{12}
(I_{N_2}+{\rho }_{12}H_{12}H^{\dagger }_{12})^{-1}
\sqrt{{\rho }_{11}{\rho }_{12}}H_{12}H^{\dagger }_{11})\nonumber\\
&\stackrel{(a)}{=}& f\left(\left(N_2,\left({\alpha }_{22},M_2\right),\left({\alpha }_{12},M_1\right)\right)\right)+{\alpha }_{11}{\min  \left({\left(M_1-N_2\right)}^+,N_1\right)\ }+\nonumber\\
&&{\left({\alpha }_{11}-{\alpha }_{12}\right)}^+\left({\min  \left(M_1,N_1\right)\ }-{\min  \left({\left(M_1-N_2\right)}^+,N_1\right)\ }\right)
+o({\log \mathsf{SNR}\ }),
\end{eqnarray}
where $(a)$ is obtained from Lemma \ref{dof2} and Lemma \ref{lemma-dof}. Now, dividing both sides by $\log \mathsf{SNR}$, the fifth GDoF bound results.

\noindent {\bf Sixth term:} The sixth term is similar to the fifth term by exchanging $1$ and $2$ in the indices.

\end{appendices}

\bibliographystyle{IEEETran}
\bibliography{bib}

\end{document}